\documentclass[11pt,draftcls,onecolumn]{IEEEtran} 

\usepackage{cite}
\usepackage{pslatex}
\usepackage{moreverb}
\usepackage{epsfig}
\usepackage{amsmath}
\usepackage{amssymb}
\usepackage{booktabs}
\usepackage{array}
\usepackage{multirow}
\usepackage{threeparttable}
\usepackage{url}
\usepackage{dsfont}

\newtheorem{lemma}{Lemma}
\renewcommand{\QED}{\QEDopen}

\begin{document}

\title{Peak to Average Power Ratio Reduction for Space-Time Codes That Achieve Diversity-Multiplexing Gain Tradeoff}

\author{
    \IEEEauthorblockN{Chung-Pi Lee and Hsuan-Jung Su*}\\
    \IEEEauthorblockA{Graduate Institute of Communication Engineering\\
    Department of Electrical Engineering\\
        National Taiwan University, Taipei, Taiwan\\
        Email: r94942026@ntu.edu.tw, hjsu@cc.ee.ntu.edu.tw}
\thanks{The material in this paper was presented in part at the Annual Conference on Information Sciences and Systems (CISS), Princeton. New Jersey, Mar. 2008, and the IEEE International Symposium on Personal, Indoor and Mobile Radio Communications (PIMRC),
Cannes, France, Sept. 2008.}
}
\newtheorem{thm}{Theorem}
\date{}
\maketitle

\begin{abstract}
Zheng and Tse have shown that over a quasi-static channel, there
exists a fundamental tradeoff, known as the diversity-multiplexing
gain (D-MG) tradeoff. In a realistic system, to avoid inefficiently
operating the power amplifier, one should consider the situation
where constraints are imposed on the peak to average power ratio
(PAPR) of the transmitted signal. In this paper, the D-MG
tradeoff of multi-antenna systems with PAPR constraints is analyzed.
For Rayleigh fading channels, we show that the D-MG tradeoff
remains unchanged with {\it any} PAPR constraints larger than one. This result implies that, instead of designing codes on a case-by-case basis, as done by most existing works, there possibly exist general methodologies for designing space-time codes with low PAPR that achieve the optimal D-MG tradeoff. 
As an example of such methodologies, we propose a PAPR reduction method based on constellation shaping that can be applied to existing optimal space-time codes without affecting
their optimality in the D-MG tradeoff. Unlike most PAPR reduction methods, the proposed method does not introduce redundancy or require side information being transmitted to the decoder.
Two realizations of the proposed method are considered. The first is similar to the method proposed by Kwok except that we
employ the Hermite Normal Form (HNF) decomposition instead of the
Smith Normal Form (SNF) to reduce complexity. The second takes the idea of integer
reversible mapping which avoids the difficulty in matrix decomposition when the number of antennas becomes large.
Sphere decoding is performed to verify that the proposed PAPR reduction
method does not affect the performance of optimal space-time codes.
\end{abstract}

\vspace{-4mm}
\begin{center}
   {\underline{\bf \small EDICS}} \hspace{3mm} {\small MSP-STCD}
\end{center}
\vspace{-6mm}

\section{Introduction}

The results in \cite{diva} on the diversity-multiplexing gain
(D-MG) tradeoff spurred numerous research activities towards the
construction of space-time codes achieving the optimal tradeoff
\cite{acht,anop,appu,theg,latc,pers,exps}.
%
When examining these space-time codes, we find that these codes generally
lead to high peak-to-average power ratio (PAPR) on each
antenna. In practice, PAPR of the signals transmitted is an important
parameter to be considered during hardware design. A high
PAPR poses difficulties in the design of the amplifier and raises the
cost of the transmitter. These practical issues motivate our study
on the D-MG tradeoff of multi-antenna systems with PAPR constraints.
For Rayleigh fading channels, our analytical result shows that the D-MG
tradeoff remains the same with {\it any} PAPR constraints larger than one. To the best of our knowledge, this is the first analytical result in the literature on the D-MG tradeoff of multi-antenna systems with PAPR constraints. This result implies that, instead of designing codes on a case-by-case basis, as done by most existing works (e.g., \cite{maxd}), there possibly exist general methodologies for designing space-time codes with low PAPR that achieve the optimal D-MG tradeoff.
As an example of such methodologies, we propose a PAPR reduction method based on constellation shaping that can be applied to existing optimal space-time codes without affecting
their optimality in the D-MG tradeoff. Unlike most PAPR reduction methods, the proposed method does not introduce redundancy or require side information being transmitted to the decoder.
In general, constellation shaping can be tailored to serve different purposes (e.g., minimizing the average transmission power) which often result in different shaping regions. The purposes of the proposed method are reduction of PAPR, and not affecting the {\it rate} and {\it optimality in the D-MG tradeoff} of the original code. For easier implementation and illustration, the {\it target} shaping region is a hypercube which will lead to an asymptotic PAPR of $3$ when the constellation size is large. Lower PAPR might be possible with a different shaping region, which, however, might be difficult to implement.
A similar approach was proposed in \cite[Chapter 5]{pear} for PAPR reduction of orthogonal frequency division multiplexing (OFDM) systems with constellation shaping based on the
Smith Normal Form (SNF) \cite{onsy} decomposition of integer matrices. Due to the prohibitive computational complexity of the SNF decomposition when the number of OFDM carriers is large, the author of \cite{pear} also considered discrete Hadamard transform (DHT) based multi-channel systems which rendered a low-complexity SNF decomposition. The authors of \cite{intc} then took the constellation shaping algorithm derived for DHT-based systems and applied it to OFDM systems, in conjunction with a selective mapping (SLM) method which incurred redundant bits to overcome the residual PAPR problem due to the mismatch in constellation shaping.

Two realizations of the proposed method will be discussed. The first is similar to the one in \cite{pear}, except that we
employ the Hermite Normal Form (HNF) decomposition \cite{intd}\cite{como} instead of the SNF decomposition to reduce the computational complexity. The second takes the idea of integer
reversible mapping \cite{matf}\cite{cust} which avoids the bit assignment problem in the above methods, and the difficulty in integer matrix decomposition when the size of the matrix becomes large. Therefore, this approach is more suitable for the situations where the number of transmit antennas or the number of OFDM carriers is large. Aside from these advantages over the methods in \cite{pear} and \cite{intc}, it is also worth mentioning that our work is better justified because the integer-based constellation shaping is crucial in preserving the optimality of space-time codes, while for the uncoded OFDM application considered in \cite{pear} and \cite{intc}, the integer-based constellation shaping is not necessary, and its advantage over the non-integer-based shaping schemes (for example, the single carrier frequency division multiple access (SC-FDMA) \cite{sc-fdma} scheme is equivalent to using a discrete Fourier transform (DFT) to shape the constellation) is yet for investigation. Note that the concept and derivation of the proposed method are very general, thus they can be applied to any linear transform based multi-channel modulation. For the space-time codes considered in this paper, simulation results using sphere decoding verify that the proposed PAPR reduction
method does not affect the optimality of the codes.

The rest of this paper is organized as follows. Section
\ref{systemdef} introduces the system model and the definitions of
diversity and multiplexing gains. In Section \ref{dmgwithpar}, we
analyze the D-MG tradeoff with any PAPR constraint larger than one, and show that, in Rayleigh fading channels, the
D-MG tradeoff remains unchanged. In
Section \ref{sec acs}, a unified framework of approximate cubic
shaping is described. In Section \ref{sec HNF} and Section \ref{sec
PLUS}, we propose two approaches of PAPR reduction via
approximate cubic shaping. The first selects the
transmitted signal using the HNF decomposition, while the second takes the idea of integer reversible
mapping. Section \ref{sec sr} provides some
simulation results and discussions. Finally, conclusions are drawn
in Section \ref{sec conclusion}.

\section{System Model and Definitions}
\label{systemdef}
\subsection{System Model} As in \cite{diva},
consider a wireless link with $m$ transmit and $n$ receive antennas.
The fading coefficient $h_{ij}$ is the path gain from
transmit antenna $j$ to receive antenna $i$. Let the channel matrix
${\mathbf{H}} = [h_{ij} ] \in \mathbb{C}^{n \times m}$. We assume that the fading coefficients are independent complex Gaussian with zero mean, unit variance, and known to the receiver, but not to the transmitter.
We also assume that the channel matrix $ {\bf{H}}$ remains constant
within a block of $l$ symbols. That is, the block length is much
smaller than the coherence time of the channel. Then the channel,
within one block, can be written as
\begin{equation} {\bf{Y}} = \sqrt {\frac{{SNR}}{m}} {\bf{HX}} +
{\bf{W}} \label{a1}
\end{equation}
where ${\mathbf{X}} \in \mathbb{C}^{m \times l}$ has entries $x_{ij}$,
$i = 1,...,m,j = 1,...,l$, being the signals transmitted by antenna
$i$ at time $j$ such that the average transmission power on each antenna in each symbol duration is
$1$; ${\mathbf{Y}} \in \mathbb{C}^{n \times l}$ is the
received signal; ${\mathbf{W}}$ is the additive noise with independent and identically distributed ({\it
i.i.d.}) entries $w_{ij} \sim \mathbb{C}N(0,1)$ (i.e., complex
Gaussian with mean $0$ and variance $1$); $SNR$ is the average
signal-to-noise ratio (SNR) at each receive antenna.
A codebook ${\bf{C}}$ with rate $R$  bits per second per hertz (b/s/Hz) is used, which
has $|{\mathbf{C}}| = {2^{Rl} }$ codewords each being an $m \times
l$ matrix.

\subsection{Diversity and Multiplexing Gains} For the case without PAPR
constraints on each antenna, in order to achieve a certain fraction
of the capacity at high SNR, one should consider a family of codes
that support a data rate which increases with $\log (SNR)$.
The diversity and multiplexing gains are defined as \cite{diva}
\newtheorem{defi}{Definition}
\begin{defi}
A  diversity gain $d^{*}(r)$  is  achieved at  multiplexing gain $r$
if the data rate $ R(SNR)$ satisfies
\begin{equation}
\begin{gathered}
\mathop {\lim }\limits_{SNR \to \infty } \frac{{R(SNR)}}
{{\log SNR}} = r \hfill \\
\end{gathered}
\label{a2}
\end{equation}
and the outage probability $P_{out} (R)$ satisfies
\begin{equation}
\begin{gathered}
 \mathop {\lim }\limits_{SNR \to \infty } \frac{{\log P_{out} (R)}}
{{\log SNR}} =  - d^{*}(r) \hfill \\
\end{gathered}
\label{a3}
\end{equation}
\QED
\end{defi}

The function $d^{*}(r)$ characterizes the D-MG
tradeoff.
For convenience, we borrow the notation introduced in \cite{diva} to
denote exponential equality. That is, $f(SNR) \doteq SNR^b$ means
\[
\mathop {\lim }\limits_{SNR \to \infty } \frac{{\log f(SNR)}} {{\log
SNR}} = b.
\]
$\dot  \geq$ , $\dot  \leq$ are similarly defined.

\section{Diversity-Multiplexing Gain Tradeoff with PAPR
constraints}\label{dmgwithpar} When space-time codes are used in a
multi-antenna system, due to the coding procedure which combines the
information symbols to form the coded symbols for each transmit
antenna, high PAPR values may occur, especially when the number of
transmit antennas is large. To reflect the limitations of practical
communication systems, we take PAPR into consideration and
investigate the effect of PAPR constraints on the D-MG tradeoff.
\subsection{The Behavior of Capacity at High SNR with PAPR Constraints}
For the  study on the optimal D-MG tradeoff with PAPR constraints,
characterization of the multiplexing gain is needed. That is, we
need to know how the capacity grows with SNR. However, the
expression of the exact capacity of a multi-antenna channel with
inputs subject to average total power and PAPR constraints may not
be a closed form, or may be too complicated (for the single antenna
scenario with average power and peak power constraints, see
\cite{thei}, \cite{thec}). Fortunately, since the D-MG tradeoff is an
asymptotic tradeoff, what we need is simply the behavior of the
capacity for asymptotically large SNR. In this section, we will
derive a lower bound of the capacity with average total power and
PAPR constraints. The bound is tight enough for the derivation of
the D-MG tradeoff. The capacity without PAPR constraints (already
known in \cite{capo}, \cite{lays}) can be used as an upper bound.
These two bounds are then used to characterize the capacity for
large SNR.

Since the channel remains constant within a block, the capacity achieving signal and average power distribution should not favor one symbol duration over another within the same block. Thus, for the purpose of analyzing the capacity with respect to the average SNR, it suffices to focus on any symbol duration within a block. We take the signal and noise vectors in (\ref{a1}) pertaining to the same symbol duration, and drop the time index to form a new vector channel model
\begin{equation} \mathbf{y = Hx + w}
\label{a6}
\end{equation}
where $ {\mathbf{x}} \in \mathbb{C}^m$ is the transmitted signal vector scaled by the transmission power, ${\mathbf{y}} \in \mathbb{C}^n$ is the received signal vector, and the additive noise vector ${\bf{w}}$ has {\it i.i.d.} entries $w_i
\sim \mathbb{C}N(0,1)$.
%
%
The average total power and PAPR constraints of the transmitted
signal ${\bf{x}}$ are ${\bf{P}}>0$ and $\rho _i>1$, $i=1,\ldots m$, respectively,
such that
\begin{equation}
 \mbox{Tr}\left( {E_{\mathbf{x}} \left[ {{\mathbf{xx}}^\dag  }
\right]} \right) \leq {\mathbf{P}}, \label{a7}
\end{equation}
\begin{equation}
\;\;\;\;\;\;\;\;\;\;\;\frac{|x_i |^2 } {{E_{x_i } [|x_i |^2 ]}}
\leq \;\rho_i, \;\;\; i=1,\ldots m, \label{a8}
\end{equation}
where $\mbox{Tr}()$ denotes trace and $\textbf{x}^{\dagger}$ denotes
the conjugate transpose of $\textbf{x}$, $E_{t}[{\;}]$ denotes
the expectation with respect to the distribution of $t$, and $x_i$ is the $i$-th element of
${\bf{x}}$. With these definitions, we
have the following lower bound on the capacity of this channel.
\begin{lemma}
The ergodic capacity $C$ of the channel (\ref{a6}) with the transmitted
signal subject to (\ref{a7}), (\ref{a8}) is:
\begin{equation}
C\geq E_{\mathbf{H}} \left[ {\log \det \left( {{\mathbf{I}} +
\frac{{\mathbf{P}}} {m}  {\mathbf{H}}{\mathbf{H}}^\dag} \right)}
\right] + \sum\limits_{i = 1}^m {k_i }\label{lem1eq}
\end{equation}
where $k_i$ are constants defined in Appendix \ref{appxa}.
\end{lemma}
\begin{proof}
The proof is given in Appendix \ref{appxa}.
\end{proof}
\subsection{Optimal D-MG Tradeoff with PAPR Constraints}
Now we are ready to discuss the D-MG tradeoff with PAPR constraints.
We have
\begin{lemma} \label{lem2}
For the Rayleigh fading channel, the ergodic capacity $C$ of the
channel (\ref{a6}) with transmitted signal subject to (\ref{a7}),
(\ref{a8}) is

\begin{equation}
C \doteq \min (m,n)\log (SNR).
\end{equation}
\end{lemma}
\begin{proof}
Let $C_\infty$ be the capacity without PAPR constraints. It is well
known that \cite{capo}\cite{lays}
\[
\begin{split}
 C_\infty = E_{\mathbf{H}} \left[ {\log \det \left(
{{\mathbf{I}} + \frac{{\mathbf{P}}} {m}
{\mathbf{H}}{\mathbf{H}}^\dag} \right)} \right]   \doteq \min
(m,n)\log(SNR).
\end{split}
\]
Using $C_\infty$ as an upper bound, from (\ref{lem1eq}), we have
 \[
C_\infty + \sum\limits_{i = 1}^m {k_i } \leq C \leq C_\infty
                                      \]
and clearly,
\[
    C_\infty +
\sum\limits_{i = 1}^m {k_i }  \doteq \min (m,n)\log (SNR).
\]
Thus,
\[
C \doteq \min (m,n)\log (SNR).
\]
\end{proof}
Lemma \ref{lem2} shows that  the multiplexing gain $r$ remains the same even with PAPR constraints. The
main result is given in the following theorem.
\begin{thm}
For the Rayleigh fading channel, the optimal D-MG tradeoff with any
PAPR constraint $\rho>1$ is  the same as the case without PAPR
constraints .
\end{thm}
\begin{proof}
The outage probability is
\begin{align}
P_{out} (R) &= \mathop {\min }\limits_{f_{\mathbf{x}}
({\mathbf{x}})} P\left[ {I({\mathbf{x}};{\mathbf{y}}|{\mathbf{H}}) <
R}
\right]\notag\\
 &\leq P\left[ {\log \det \left( {{\mathbf{I}} + \frac{{\mathbf{P}}}
{m}  {\mathbf{H}}{\mathbf{H}}^\dag} \right) + \sum\limits_{i = 1}^m
{k_i }  < R} \right]
\notag\\
 &\doteq P\left[ {\log \det \left( {{\mathbf{I}} + SNR  {\mathbf{H}}{\mathbf{H}}^\dag} \right) + \sum\limits_{i = 1}^m {k_i }  < R} \right]
\label{a331}
\end{align}
where $I({\;};{\;})$ denotes the mutual information and
$f_{\mathbf{x}} ({\mathbf{x}})$ is the probability density function
of ${\mathbf{x}}$ subject to equations (\ref{a7}) and (\ref{a8}).
The inequality follows from (\ref{a29}) and (\ref{a331}) follows
from equation (9) in\cite{diva}. Using the same techniques as in
\cite{diva}, denoting $\lambda _i$ as the nonzero eigenvalues of
${\mathbf{HH}}^\dag$ and letting $R = r\log SNR$, $\sum\limits_{i =
1}^m {k_i }  = K$, $\lambda _i  = SNR^{ - \alpha _i }$, $(x)^ +
\triangleq \max (x,0)$, we have
\begin{align}
&P\left[ {\left( {\log \det \left( {SNR{\mathbf{H}} {\mathbf{H^\dag
+ I}}} \right) + K} \right) < r\log SNR} \right]
\notag\\
 = &P\left[ \left( {\prod\limits_{i = 1}^n {(1 + SNR\lambda _i )}  } \right) <  \frac{SNR^{r}}{e^{K}} \right]\;
\notag\\
 \doteq &P\left[ {\sum\limits_{i=1}^n {(1 - \alpha _i )^ +  }  < r} \right].\;
\label{a32}
\end{align}
Thus $P_{out} (R)\dot \leq$ (\ref{a32}) and (\ref{a32}) is
exponentially equal to the outage probability without PAPR
constraints in \cite{diva}. However, the outage probability with PAPR
constraints should be larger than the outage probability without PAPR
constraints, that is, $P_{out} (R)\dot \geq$ (\ref{a32}). Thus
$P_{out} (R) \doteq$ (\ref{a32}), and the optimal tradeoff remains
the same as the case without PAPR constraints.
\end{proof}
Intuitively, this result is not surprising, since the PAPR
constraints do not reduce the spatial degree of freedom and the
capacity ${C}$ grows like $C_\infty$ with increasing SNR.

To show that this optimal tradeoff can be achieved by a code with
finite code length, we adopt a similar method as in \cite{diva} by
choosing the input to be a random code drawn from {\it i.i.d}
distribution (\ref{1b}).
\begin{thm}  \label{thm2}
For $l \geq m + n - 1$, in Rayleigh fading channels with any PAPR
constraint $\rho>1$, the optimal D-MG tradeoff is achievable.
\end{thm}
\begin{proof}
The proof is given in  Appendix \ref{appxc}.
\end{proof}


\section{Approximate Cubic Shaping}\label{sec acs}
In this section, we discuss the fundamental concepts of the shaping techniques we use to reduce the PAPR of existing D-MG optimal space-time codes. A constellation generally consists of a set of points on an
$l$-dimensional complex lattice, or an $L$-dimensional real lattice
$\lambda _L$ (where $L=2l$), that are enclosed within a finite
region $\xi _L$. The boundary of a signal constellation affects the average power and PAPR for a given transmitted
data rate.
In selecting the signal constellation, one tries to minimize the
average power with low PAPR. The $L$-dimensional constellation
consisting of all the points enclosed within an $L$-dimensional cube
is called cubic shaping, which leads to a PAPR value equal to $3$
when the constellation size approaches infinity. With the same
number of points to be transmitted, the reduction in the average
transmission power due to the use of a region $\xi _L$ as signal
constellation instead of a hypercube is referred to as the shaping
gain $\eta _s$ of $\xi _L$.  The region that has the smallest
average power for a given volume is an $L$-dimensional sphere.
Although the sphere shaping gives the best shaping gain, it also
results in high PAPR values when $L$ is large. Shaping of
multidimensional constellation has been extensively studied
previously \cite{mulc,tres,sham,onop}.
For our interest in PAPR reduction, we will focus on the cubic
shaping due to its good PAPR value asymptotically equal to $3$.

Consider the shaping on a general space-time code ${\mathbf{X}}$ in
the form of
\begin{equation}
{\mathbf{x}} = {\mathbf{Gs}},\;\;\;\;\;\;{\mathbf{s}} \in
\mathbb{Z}^M,{\mathbf{G}} \in \mathbb{R}^{M \times M}\label{a31}
\end{equation}
where ${\bf{x}}$ is an isomorphic vector representation of
${\bf{X}}$, ${\bf{G}}$ is an invertible generator matrix and
${\bf{s}}$ is the vector of information symbols chosen from
$M$-dimensional integer lattice $\mathbb{Z}^M$. A Quadrature
Amplitude Modulation (QAM) constellation is a subset of a
scaled integer lattice $\mathbb{Z}^M$.

For example, an $m\times m$ D-MG optimal
space-time code ${\textbf{X}}$ proposed in \cite{pers} can be
expressed in terms of $m$ vectors $\textbf{x}^{(i)}$,
$i=1,2,...,m$,
 \[\textbf{ x}^{(i)}= {\textbf{G}^{(i)}}\textbf{ s}^{(i)}\]
 \[\textbf{s}^{(i)}\in {\mathbb{Z}[i]}^{m}, \; \textbf{G}^{(i)}\in {\mathbb{C}^{m\times m}}      \]
where $\mathbb{Z}[i]$ stands for the Gaussian integers (i.e. $a+bi,
a, b\in \mathbb{Z}$) and each $\textbf{ x}^{(i)}$ corresponds to the
symbols in the space-time codeword matrix with positions corresponding to the nonzero elements' positions of
$\textbf{B}^{i-1}$, where
\[
 \textbf{B}=
\begin{pmatrix}
  0 &  0 & \cdots & 0 & 1  \\
  1 &  0 & \cdots & 0 & 0  \\
  0 &  1 & \cdots & 0 & 0  \\
    &    &  \vdots &  &    \\
  0 & 0  & \cdots &1 &0
\end{pmatrix}.
\]
Note that although these symbols are on different time and different
antennas, they have equal average power with respect to all
codewords owing to the code structure \cite{pers}.
 For each $\textbf{ s}^{(i)}$,
$\textbf{G}^{(i)}$, we can get the isomorphic representation by
separating the real and imaginary parts of $\textbf{ s}^{(i)}$ and
$\textbf{G}^{(i)}$ as follows,
\[
{\textbf{ s}'^{(i)}}  =
\begin{pmatrix}
\textbf{s}^{(i)}_{Re} \\
\textbf{s}^{(i)}_{Im}
\end{pmatrix}
\]
and
\[
 \textbf{G}'^{(i)}=
\begin{pmatrix}                
  \textbf{G}^{(i)}_{Re} &  -\textbf{G}^{(i)}_{Im}  \\
   \textbf{G}^{(i)}_{Im} &   \textbf{G}^{(i)}_{Re}
\end{pmatrix}.
\]
Then $\textbf{ x}'^{(i)}= {\textbf{G}'^{(i)}}\textbf{ s}'^{(i)}$,
where $\textbf{s}'^{(i)}\in {\mathbb{Z}}^{2m}, \;
\textbf{G}'^{(i)}\in {\mathbb{R}^{2m\times 2m}}$. This is exactly
the same form as (\ref{a31}).

Our goal is to shape the transmitted signals such that the
constellation region of $\bf{x}$ is cubic. However, as the constellation points of $\bf{x}$ have to be on the lattice that achieves the optimal D-MG tradeoff, the constellation region will not be exactly cubic. The idea is to shape by cosets. Since the information symbol $\bf{s}$'s lattice is more regular (integer lattice) and easier to label, cosets will be found in the domain of $\bf{s}$ using a basis corresponding (approximately) to the cubic basis for $\bf{x}$. The following two steps illustrate how this can be done.

\emph{\textbf{Step 1}}: Introduce a set of perturbation vectors $U$
such that each $\mathbf{u} \in U$ is a linear combination of vectors
${\mathbf{v}}^i$, $i=1 \ldots M$
\begin{equation}
{\mathbf{u}} = \alpha _1 {\mathbf{v}}^1  + \alpha _2 {\mathbf{v}}^2
+ ... + \alpha_M {\mathbf{v}}^M \label{a33}
\end{equation}
where ${\mathbf{\alpha }} = [\alpha _1 ,\alpha _2 ,...,\alpha _M ]^T
\in \mathbb{Z}^M$ and ${\mathbf{v}}^i$ has the properties that
\begin{equation}
\begin{split}
\;\;{\mathbf{Gv}}^i  &\triangleq {\mathbf{\bar v}}^i, \;\;\;\;i =
1,2,...,M\\
{\mathbf{\bar v}}^i  &= \left[ {\varepsilon _1,\varepsilon _2
,...,\mu _i,...,\varepsilon _M } \right]^T,\;\;|\mu _i | >  >
|\varepsilon _j |, \label{a34}
\end{split}
\end{equation}
and
\begin{equation}
|\mu _1 | \cong |\mu _2 | \cong ... \cong |\mu _M | . \label{a36}
\end{equation}
Let $(\mathbf{s}+\mathbf{u}), \forall \mathbf{u} \in U$, be the coset representing
the same information.

\emph{\textbf{Step 2}}: Choose $(\mathbf{s}+\mathbf{u^{*}})$ as the
vector of information symbols such that the transmitted signals,
${\mathbf{x}} = {\mathbf{G}}({\mathbf{s}} + {\mathbf{u}}^{*})$
consist of an approximate cubic constellation.  The possible
transmitted signals can be written as
\begin{align}
{\mathbf{G}}({\mathbf{s}} + {\mathbf{u}}) &= {\mathbf{Gs}} +
{\mathbf{Gu}} = {\mathbf{\bar s}} + {\mathbf{\bar u}}\; \notag\\
 &= {\mathbf{\bar s}} + (\alpha _1 {\mathbf{\bar v}}^1  + \alpha _2 {\mathbf{\bar v}}^2  + ... + \alpha _M {\mathbf{\bar v}}^M )
 \label{a35}
\end{align}
where ${\mathbf{\bar s}} \triangleq {\mathbf{Gs}}$, ${\mathbf{\bar
u}} \triangleq {\mathbf{Gu}}$. In this particular set $U$, each
${\mathbf{u}}$ causes relatively large perturbations on certain
elements of ${\mathbf{\bar s}}$ where the corresponding $\alpha _i
\ne 0$. If we treat $ \varepsilon _j$'s as $0$, to put $\mathbf{x}$
in the cubic constellation, ${\mathbf{u}}^{*}$ can be searched
accordingly by modulo operations. However, the mapping from $\bf{s}$
to $\bf{x}$ has to be reversible for successful decoding. In other
words, these approximations need to be reversible. In the following
two sections, we will propose two such mappings for approximate
cubic shaping.
An approximate hypercube
constellation leads to a low PAPR value of 3 when  $ \varepsilon
_j$'s  are relatively small, i.e., when the constellation is large
enough.
\section{approximate Cubic Shaping via Hermite Normal Form (HNF)
Decomposition}\label{sec HNF}
Firstly, we need to decide
${\mathbf{v}}^i,\;i = 1,2,...,M$, in (\ref{a33}). Consider a
partition ${\text{ }}\mathbb{Z}^M /\Lambda$, where the lattice $
\Lambda  = {\mathbf{Q}}\mathbb{Z}^M$, and ${\bf{Q}}$ is an $M
\times M$ {\it integer} matrix such that
\[
{\text{ }}\;\;\;{\mathbf{GQ}} \cong \sigma
{\mathbf{I}}.\;\;\;\;\;\;
\]
The approximation is due to the fact that $\mathbf{Q}$ is an integer matrix. If we
choose ${\mathbf{v}}^i$ as
\begin{align}
{\mathbf{v}}^i  &= {\mathbf{Qe}}^i \; \\
{\mathbf{e}}^i \; &= [0_{(1)},...,0_{(i-1)}, 1_{(i)}, 0_{(i+1)},...,0_{(M)}
]^T, \notag
\end{align}
clearly, ${\mathbf{v}}^i$ has the properties (\ref{a34}) and
(\ref{a36}) when $\sigma$ is reasonably large. Thus,
\begin{align}
{\mathbf{\bar v}}^i  &\triangleq {\mathbf{Gv}}^i \; =
{\mathbf{GQe}}^i \notag\\
 &= \;\left[ {\varepsilon _1 ,\varepsilon _2 ,...,\sigma _i ,...,\varepsilon _M } \right]^T ,\;\;|\sigma _i | >  > |\varepsilon _j |
\label{a38}\\
 &\cong \sigma {\mathbf{e}}^i .
\notag
\end{align}
Define $U \triangleq {\mathbf{Q}}\mathbb{Z}^M$. We can rewrite
(\ref{a35}) in terms of  coset ${\mathbf{s}} + U$
\begin{align}
{\mathbf{G}}{\text{(}}{\mathbf{s}} + U ) &=
{\mathbf{Gs}} + {\mathbf{GQ}}\mathbb{Z}^M \notag\\
 &= {\mathbf{\bar s}} + {\mathbf{GQ}}\mathbb{Z}^M
\notag\\
 &\cong {\mathbf{\bar s}} + \sigma \mathbb{Z}^M. \label{a39}
\end{align}
Approximate cubic shaping can be done by treating $\varepsilon _j$'s
as $0$ (equivalently, the approximation in (\ref{a39}) as equality),
then searching for ${\mathbf{u}}^{*} \in \textbf{Q}
\mathbb{{Z}}^{M}$ to put ${\mathbf{x}}$ in the approximate cubic
constellation.

A geometric interpretation of this shaping method is that we choose
${\mathbf{s}} + {\mathbf{u}}^{*}$ in a shaped constellation whose
boundary is a parallelotope defined along the columns of ${\bf{Q}}$.
Thus the signal boundary in the domain of ${\bf{x}}$ translates to
an approximate hypercube. In the following, we will describe the
shaping process in three parts: \emph{(1)}\emph{ determine}
$\mathbf{Q}$ \emph{(2)} \emph{find the coset leaders}
\emph{(3)}\emph{put} $\mathbf{x}$ \emph{in an approximate
hypercube}, which are derived in a different point of view from
similar works proposed in\cite{intc} and \cite{pear}. Note that
\cite{intc} and \cite{pear} deal with the PAPR of single-antenna
OFDM systems but not
D-MG optimal space-time coded systems whose transmitted signals need to be on certain lattices.

\textbf{\emph{(1)}\emph{ Determine}} $\mathbf{Q}$: The
number of cosets, $|\det ({\mathbf{Q}})|$ (which will manifest in
part (2)), must be large enough to support the target number of
points we want to transmit. Therefore, let
\[
\begin{gathered}
  \;\;\;\;\;\;\;\;\;\;\;\;\;\;\;\;{\mathbf{Q}} = \left[ {\tilde \sigma {\mathbf{G}}^{ - 1} } \right] \hfill \\
  \;\;\;\;\;\;\;\;\;\;\;\;\;\;\;\;\;|\det ({\mathbf{Q}})| \geq \sigma ^M  \hfill \\
\end{gathered}
\]
where $\left[ {\;} \right]$ denotes rounding, which makes the set of
perturbation vectors ${\bf{u}}$ belong to the integer lattice
$\mathbb{Z}^M$,
and $|\det ({\mathbf{Q}})|$ is the volume of the
parallelotope defined by ${\bf{Q}}$, or equivalently, the number of
points in the parallelotope.  $\sigma^M$ is the number of
transmitted points. The parameter $\tilde \sigma$ should be chosen
to be the smallest value that ensures the number of points in the
shaped constellation larger than the number of points in the
unshaped constellation, so no information will be lost. For the case we concern, $\mathbf{Q}$ can
always be chosen as a nonsingular matrix when $\tilde \sigma$ is
large enough.

\textbf{\emph{(2)} \emph{Find the coset leaders} }: The coset
leaders ${\bf{s}}$  must satisfy
\begin{equation}
\begin{gathered}
  \;\;\;\;\;\;\;\mbox{if}\;\;\;{\mathbf{s}}^i  \ne {\mathbf{s}}^j \;\;\;\;\;\;\;\;\;\;\;\;\;\; \hfill \\
  \;\;\;\mbox{then}\;\;\;{\mathbf{s}}^i  \ne {\mathbf{s}}^j  + {\mathbf{Qz}},\;\;\;\;\forall {\mathbf{z}} \in \mathbb{Z}^M
\label{a40}\end{gathered}
\end{equation}
where ${\mathbf{s}}^i ,\;{\mathbf{s}}^j$ are coset-leaders of two
different cosets $\;{\mathbf{s}}^i  + {\mathbf{Q}}\mathbb{Z}^M$,
$\;{\mathbf{s}}^j  + {\mathbf{Q}}\mathbb{Z}^M$, respectively, so
there is no ambiguity in decoding. As an example, consider the simplest case
when
\[
{\mathbf{Q}} = {\mathbf{D}} = {\text{diag}}(d_1 ,d_{_2 } ,...,d_M ).
\]
Denoting ${\bf{S}}$ as the set of coset leaders, it is natural to
choose ${\bf{S}}$ as
\begin{equation}
  {\mathbf{S}} \triangleq \{ {\mathbf{s}}|\;0 \leq s_i  < d_i \;,\;i = 1,2,...,M\}
\label{a41}
\end{equation}
where ${\mathbf{s}} = [s_1 ,s_2 ,...,s_M]^T$. Obviously, the
coset-leaders ${\mathbf{s}} \in {\mathbf{S}}$ satisfy (\ref{a40})
and $\mathbf{S}$ contains all the coset leaders. The number of coset
leaders is equal to $|\det ({\mathbf{D}})|$. For example,
\[
\mbox{if}\;\;{\mathbf{D}} = \left( {\begin{array}{*{20}c}
   1 & 0 & 0  \\
   0 & 2 & 0  \\
   0 & 0 & 3  \\

 \end{array} } \right)\;\mbox{then}\;\;{\mathbf{S}} = \left\{
 \begin{gathered}
 {}[0,0,0]^T ,\;[0,0,1]^T ,[0,0,2]^T  \hfill \\
  [0,1,0]^T ,\;[0,1,1]^T ,[0,1,2]^T  \hfill \\
\end{gathered}  \right\}
\]
and the number of coset-leaders in ${\bf{S}}$ is $\det
({\mathbf{D}}) = 6$.

For the general case when ${\bf{Q}}$ is not a diagonal matrix,
decompose ${\bf{Q}}$ into
\begin{equation}
{\mathbf{Q}} = {\mathbf{UDV}}
\end{equation}
where ${\mathbf{U,V}}$ are unimodular matrices (i.e. integer
matrices with $|\det ({\mathbf{U}})| = 1,{\text{ }}|\det
({\mathbf{V}})| = 1$).  The matrix ${\bf{D}}$ is called the Smith
Normal Form (SNF) of the matrix ${\bf{Q}}$ \cite{onsy}. We can first
index the coset leaders as (\ref{a41}), and left-multiply ${\bf{s}}$
by $\textbf{U}$ such that $ {\mathbf{Us}}$ is the coset leader of
${\mathbf{Us}} + {\mathbf{UD}}\mathbb{Z}^M$. Define
\begin{equation}
{\mathbf{S}}_{\mathbf{U}}  \triangleq \{ {\mathbf{Us}}|\;0 \leq s_i
< d_i,\;i = 1,2,...,M\}. \label{snfindex}
\end{equation}
Since ${\mathbf{U}}$ is a unimodular matrix, $
{\mathbf{S}}_{\mathbf{U}}$ contains all coset leaders of $
{\mathbf{Us}} + {\mathbf{UD}}\mathbb{Z}^M$, and
\begin{align}
  {\mathbf{Us}} + {\mathbf{UD}}\mathbb{Z}^M  &= {\mathbf{Us}} + {\mathbf{UD}}({\mathbf{V}}\mathbb{Z}^M ) \notag\\
   &= {\mathbf{Us}} + {\mathbf{Q}}\mathbb{Z}^M \;\label{a44}
\end{align}
where the second equality follows from the fact that the lattice
${\mathbf{V}}\mathbb{Z}^M$ is identical to the lattice
$\mathbb{Z}^M$ when ${\mathbf{V}}$ is a unimodular matrix. Thus, $
{\mathbf{S}}_{\mathbf{U}}$ contains all coset-leaders of $
{\mathbf{Us}} + {\mathbf{Q}}\mathbb{Z}^M$.

The SNF decomposition can be performed via column and row
operations, which generally have the problem of intermediate
expression swell. One can use modular arithmetic to control
expression swell \cite{como}.

After examining the above algorithms, we find that the
diagonalization of SNF decomposition is not necessary. Instead, we
can decompose ${\bf{Q}}$ as
\begin{equation}
{\mathbf{Q}} = {\mathbf{RV}}
\end{equation}
where ${\bf{V}}$ is unimodular and ${\bf{R}}$ is an integer lower
triangular matrix. There is a theorem that guarantees the
existence of the decomposition of ${\mathbf{Q}} = {\mathbf{RV}}$,
known as the Hermite Normal Form (HNF) \cite{intd}. The theorem is
stated here for completeness.
\\
\begin{thm}
Any $M \times M$ invertible integer matrix ${\bf{Q}}$  can be
decomposed into ${\mathbf{Q}} = {\mathbf{RV}}$, where ${\bf{V}}$ is
a unimodular matrix and ${\bf{R}}$ is an integer lower triangular
matrix.
\\
\end{thm}

Let $r_{ii}  \ne 0$ be the diagonal elements of ${\bf{R}}$. Then we
can form the set of coset-leaders, ${\mathbf{S}}$ as
\begin{equation}
{\mathbf{S}} = \{ {\mathbf{s}}|\;0 \leq s_i  < r_{ii} \}.\label{a50}
\end{equation}
The validity of this set of coset-leaders can be verified by the
following theorem.\\
\begin{thm} \label{theorem 2}
Given a matrix $\mathbf{Q}=\mathbf{R}\mathbf{V}$, the set
$\mathbf{S}$ defined in (\ref{a50}) contains all the coset leaders
of $\mathbf{s} + {\mathbf{Q}}\mathbb{Z}^M$.
\\
\end{thm}

\begin{proof}
From (\ref{a44}), the coset leaders of $\mathbf{s} +
{\mathbf{Q}}\mathbb{Z}^M$ are the coset leaders of ${\mathbf{s}} +
{\mathbf{R}}\mathbb{Z}^M$ since
\begin{align}
{\mathbf{s}} + {\mathbf{Q}}\mathbb{Z}^M  &= {\mathbf{s}} +
{\mathbf{R}}({\mathbf{V}}\mathbb{Z}^M )\notag\\
&= {\mathbf{s}} + {\mathbf{R}}\mathbb{Z}^M.
\end{align}
To show that each $\mathbf{s} \in \mathbf{S}$ is a valid coset leader, we
need to prove  that  for ${\mathbf{s}}^i ,\;{\mathbf{s}}^j \in
{\mathbf{S}}$, %
\begin{align*}
\mbox{if}\;\;\;{\mathbf{s}}^i  &= {\mathbf{s}}^j  + {\mathbf{Rz}}\;,\;{\mathbf{z}}\in \mathbb{Z}^M\\
\mbox{then}\;\;\;\;{\mathbf{s}}^i  &= {\mathbf{s}}^j .
\end{align*}
The proof goes by induction. Let ${\mathbf{z}} = [z_1 ,z_2 ,...,z_M
]^T$, $r_{ij}$ be the entries of ${\bf{R}}$. Note that from
(\ref{a50}), if $ {\mathbf{s}}^i  = {\mathbf{s}}^j + {\mathbf{Rz}}$,
then $s_1^i = s_1^j$, $z_{1}=0$. Suppose  $s_k^i = s_k^j$ for
$k=1,2,...m-1$ and $z_k  = 0$. Then
\begin{align*}
s_m^i  &= s_m^j  + \sum\limits_{k = 1}^m {z_k r_{mk} }=
s_m^j+z_m r_{mm} \\
 &= s_m^j
\end{align*}
which completes the induction. Finally, $\mathbf{S}$ contains all the
coset leaders since
$|\text{det}(\mathbf{Q})|=|\text{det}(\mathbf{R})|$. This completes
the proof.
\\
\end{proof}

\textbf{\emph{(3)}\emph{ Put}} $\mathbf{x}$ \textbf{\emph{in an
approximate hypercube}}: Since the coset leader in $\mathbf{S}$ is
not necessarily in the parallelotope enclosed by the columns of $
{\mathbf{Q}}$. We need to do the modulo-${\bf{Q}}$ operation in
(\ref{a51}) to put ${\mathbf{\tilde s}}$ in the shaped constellation
and transmit ${\mathbf{x}} = {\mathbf{G\tilde s}}$.
\begin{equation}
\begin{gathered}
  {\mathbf{\gamma }} = \left\lfloor {{\mathbf{Q}}^{ - 1} {\mathbf{s}}} \right\rfloor  \hfill \\
  {\mathbf{\tilde s}} = {\mathbf{s}} - {\mathbf{Q\gamma }} \hfill \\
\end{gathered}\label{a51}
\end{equation}
where $\left\lfloor {\;} \right\rfloor$ denotes the floor function.
As a side note, it is desirable to translate $ {\mathbf{\tilde s}}$ to
minimize the transmit power (i.e., make $E\left[ {\mathbf{x}}
\right] \cong 0$). In this paper, however, we only concern the shape
of the constellation.

Now we summarize the algorithm using HNF decomposition as
follows:\\
\textbf{\emph{Encoding }}: Let ${\bf{s}}$ defined in (\ref{a50})
be the canonical representation of an integer $I$ which represents
the data to be sent. $\bf{s}$ can be obtained by the following recursive
modulo operation
\begin{equation}
\begin{split}
s_1  &= I\;\bmod \;r_{11}\\
I_1  &= \frac{{I - s_1 }} {{r_{11} }}\\
s_i  &= I_{i - 1} \bmod \;r_{ii}\\
I_i  &= \frac{{I_{i - 1}  - s_i }} {{r_{ii} }}\;
\end{split}\label{a46}
\end{equation}
where $2 \leq i \leq M$. Then use the algorithm defined in
(\ref{a51}) and transmit ${\mathbf{x}} = {\mathbf{G\tilde s}}$.
\\
\emph{\textbf{Decoding}} : First, an estimate of $\mathbf{\tilde s}$ is obtained from the received
signal (using, e.g., sphere demodulation). Let ${\mathbf{r}}_i$ be the $i$-th column
of ${\bf{R}}$. The decoding algorithm can be arranged to be top-down
\begin{equation}
\begin{split}
s_1  &= \;\tilde s_1 \;\bmod \;r_{11} \;\;\;\;(s_1  = \;\tilde s_1
+\;q_1 r_{11} )\\
\mbox{for} \;\; i &= 2\;\;:\;\;{\text{M}}\\
{\mathbf{\tilde s}} &= \;{\mathbf{\tilde s}} + q_{i - 1}
{\mathbf{r}_{i-1}}\\
s_i  &= \;\tilde s_i \;\bmod \;r_{ii} \;\;\;\;(s_i  = \;\tilde s_i +
\;q_i r_{ii} )\\
&{\text{end}}\;
\end{split}\label{a52}
\end{equation}
%
%

%

Compared to the similar approaches proposed in \cite{pear}, our
method can save the multiplication of ${\mathbf{Us}}$ in
(\ref{snfindex}) in encoding, and half of the multiplications in
decoding due to the lower triangular matrix, although both schemes
have the same order of complexity $O(M^2 )$. Moreover, sometimes
${\bf{U}}$ may have exceedingly large entries. Our scheme only
requires ${\bf{R}}$ and it is more efficient to only compute
${\bf{R}}$ \cite{como}.
\section{approximate Cubic Shaping via Integer Reversible Matrix
Mapping}\label{sec PLUS} In practical communication systems, the
number of points in the constellation usually equals to a number that can be expressed by an integer number of bits. That is, the constellation has $2^K$ points, where $K$ is a positive integer. In Section \ref{sec HNF}, we chose ${\mathbf{Q}} = \left[
{\tilde \sigma {\mathbf{G}}^{ - 1} } \right] $ to ensure that
${\bf{Q}}$ is an integer matrix. However, $|\det ({\mathbf{Q}})|$
is generally not in the form of $ 2^K$ due to the rounding
operation. This leads to the inconvenience of using large
$I$ in the encoding procedure (\ref{a46}), since it can not be
expressed in terms of bits. To avoid this problem, we relax the
integer constraints on the entries of ${\bf{Q}}$ and consider a
nonlinear mapping. Let the unshaped constellation be a hypercube,
namely
\begin{equation}
\;{\mathbf{S}} = \{ {\mathbf{s}}|\;0 \leq s_i  < \sigma, \forall i \}.
\label{qams}
\end{equation}
Clearly, the total number of transmitted points is $\sigma ^M$ and
we can choose $\sigma = 2^{\left( {K/M} \right)} $. Transform
${\bf{S}}$ into a shaped constellation $ {\mathbf{S}}_{\mathbf{Q}}$
\begin{equation}
{\mathbf{S}}_{\mathbf{Q}}  = \{ {\mathbf{Qs}}|\;0 \leq s_i < \sigma, \forall i
\}. \label{a53}
\end{equation}
where ${\mathbf{Q}} = {\mathbf{G}}^{ - 1}$ and $|\det
({\mathbf{Q}})|$ is normalized to $1$. Then the ${\bf{x}}$-domain
shaped constellation $ {\mathbf{GS}}_{\mathbf{Q}}$ is transformed
back to a hypercube
\[
{\mathbf{GS}}_{\mathbf{Q}}  = \{ {\mathbf{GQs}} = {\mathbf{s}}|\;0
\leq s_i  < \sigma, \forall i \} .
\]
The problem of (\ref{a53}) is that ${\mathbf{Qs}} \notin
\mathbb{Z}^M$. This will destroy the optimality of the transmitted signal. Naturally,
one method to try is
\begin{equation}
\left[ {{\mathbf{S}}_{\mathbf{Q}} } \right] = \{ \left[
{{\mathbf{Qs}}} \right]|\;0 \leq s_i  < \sigma, \forall i \} . \label{a54}
\end{equation}
However, there is a possibility that, for ${\mathbf{s}}^i, {\mathbf{s}}^j  \in {\mathbf{S}}$,
\begin{equation}
  {\mathbf{s}}^i  \ne {\mathbf{s}}^j \;\;\;\;
  {\text{but}}\;\;\;\left[ {{\mathbf{Qs}}^i } \right] = \left[ {{\mathbf{Qs}}^j }
  \right].
\end{equation}
To resolve the ambiguity, we choose an integer to integer reversible
mapping\cite{matf}, through which valid shaped symbols can be found.
Furthermore, the shaped constellation will be similar to that using
(\ref{a54}).

Firstly we borrow some definitions from \cite{matf}. If there
exists an elementary reversible structure based on a matrix for
perfectly invertible integer implementation, the matrix is called
an elementary reversible matrix (ERM). Consider an upper or lower
triangular matrix $ \mathbf{A}$ whose diagonal elements are $j_i
= \pm 1$, a reversible integer mapping is defined as follows
\cite{matf}:

Let $ \textbf{A}$ be an $M \times M$ upper triangular matrix with
elements $\{ a_{mn} \}$, and  ${\mathbf{y}} = [{\mathbf{As}}]$, that is,
\begin{equation}
\begin{split}
  y_m  &= j_m s_m  + \left[ {\sum\limits_{n = m + 1}^M {a_{mn} s_n } } \right], m = 1,2,3,..., M - 1 \hfill \\
  y_M  &= j_M s_M . \hfill
\end{split}
\end{equation}
The inverse mapping from $\mathbf{y}$ to $\mathbf{s}$ is
\begin{equation}
\begin{split}
  s_M  &= s_M /j_M  \hfill \\
  s_m  &= (1/j_m )\left( {y_m  - \left[ {\sum\limits_{n = m + 1}^M {a_{mn} s_n } } \right]} \right) \hfill \\
  m &= M - 1, M - 2,...1 .\hfill
\end{split} \label{invu}
\end{equation}
Similar results can be obtained for a lower triangular matrix. This
kind of triangular matrix is called a triangular ERM (TERM). If all
the diagonal elements of a TERM equal to $1$, the TERM will be a
unit TERM. There is  another feasible ERM form known as the single-row
ERM (SERM) with $j_m = \pm 1$ on the diagonal and only one row of
off-diagonal elements are not all zeros. The reversible integer
mapping of SERM is straightforward:
\begin{equation}
\begin{split}
  y_{m'}  &= j_{m'} s_{m'}  + \left[ {\sum\limits_{n \ne m'}^M {a_{m'n} s_n } } \right],\;\;\;\;\mbox{for}\;\;m = m' \hfill \\
  y_m  &= j_m s_m, \;\;\;\;\;\;\;\;\;\;\;\;\;\;\;\;\;\;\;\;\;\;\;\;\;\;\;\;\;\;\;\;\;\;\;\; \mbox{otherwise} \hfill \\
\end{split}
\end{equation}
where $m'$ is the row with nonzero off-diagonal elements. The inverse
operation is
\begin{equation}
\begin{split}
  s_m \; &= y_m /j_m, \;\;\;\;\;\;\;\mbox{for}\;m \ne m' \hfill \\
  s_{m'}  &= \left( {y_{m'}  - \left[ {\sum\limits_{n \ne m'}^M {a_{m'n} s_n } } \right]\;} \right)\;/\;j_{m'}. \; \hfill \\
\end{split}  \label{invs}
\end{equation}
Denote ${\mathbf{S}}_0$ as a unit SERM with $m' = M$. It has been
shown in \cite{matf} that $\bf{Q}$ has a ``PLUS" factorization.
\\
\begin{thm}
Matrix ${\bf{Q}}$  has a TERM factorization of ${\mathbf{Q}} =
{\mathbf{PLUS}}_0$ if and only if $\det ({\mathbf{Q}}) = \det
({\mathbf{P}})=\pm 1$, where ${\bf{L}}$, ${\bf{U}}$  are unit lower
and unit upper TERMs, respectively, and ${\bf{P}}$ is a permutation
matrix subject to a possible negative sign.
\\
\end{thm}
From (\ref{a53}), clearly, ${\bf{Q}}$ satisfies the property that
$\det ({\mathbf{Q}}) =  \pm 1$. Now, we summarize the shaping
algorithm using the  PLUS factorization.

\textbf{\emph{Encoding}} : In contrast to (\ref{a54}), we decompose
${\bf{Q}}$ into ${\mathbf{Q}} = {\mathbf{PLUS}}_0$ to obtain an
integer to integer reversible mapping.  The shaping algorithm is
\begin{align}
{\mathbf{\tilde s}} &= {\mathbf{P}}\left[ {{\mathbf{L}}\left[
{{\mathbf{U}}\left[ {{\mathbf{S}}_0 {\mathbf{s}}} \right]\;}
\right]\;\;} \right],\;\;\;\;\;\;\;\;\;{\mathbf{s}}\; \in
{\mathbf{S}}\;\label{a57}\\
\end{align}
where $\mathbf{S}$ is defined in (\ref{qams}). Then
${\mathbf{x}} = {\mathbf{G\tilde s}}$ is transmitted.

\emph{\textbf{Decoding}} : First, an estimate of $\mathbf{\tilde s}$ is obtained from the received
signal. Then the inverse operations (\ref{invu}),
(\ref{invs}) and $\bf{P}^{-1}$ are used to recover $\mathbf{s}$ from $\mathbf{\tilde s}$.\\

For a ${\mathbf{Q}} = {\mathbf{PLUS}}_0$ with a denotation of $
{\mathbf{e_{r}}}^{(i)}$  for the rounding error vector that results
from the transform of the $i$-th ERM, the total error due to reversible integer mapping is
\begin{equation}
 \;\;\;\;\;\; \;\;\;\;|{\mathbf{e_{r}}}| = |{\mathbf{P}}({\mathbf{e_{r}}}^{(3)}  + {\mathbf{Le_{r}}}^{(2)}  + {\mathbf{LUe_{r}}}^{(1)} )|
\label{a59}
\end{equation}
and ${\mathbf{\tilde s}} = {\mathbf{Qs}} + {\mathbf{e_{r}}}$. When
using (\ref{a57}) to shape the constellation, if we view it as a
linear operation (as the constellation becomes large, the effect of
rounding is relatively minor), we actually choose $ {\mathbf{v}}^i$
defined in (\ref{a34}) as
\begin{equation}
\begin{split}
{\mathbf{v}}^i  &= {\mathbf{P}}\left[ {{\mathbf{L}}\left[
{{\mathbf{U}}\left[ {{\mathbf{S}}_0 \sigma {\mathbf{e}}^i }
\right]\;} \right]\;\;} \right]\;\\
{\mathbf{e}}^i \; &= [0_{(1)}, ..., 0_{(i-1)}, 1_{(i)}, 0_{(i+1)}, ..., 0_{(M)}]^T. \;
\end{split}\label{a61}
\end{equation}
From (\ref{a61}),
\[
{\mathbf{Gv}}^i  = {\mathbf{\bar v}}^i  = \sigma {\mathbf{e}}^i  +
{\mathbf{Ge_r}}.
\]
Obviously, ${\mathbf{v}}^i$ satisfies the property (\ref{a34})
when $\sigma$ is large enough. Thus this method is also an
approximate cubic shaping described in Section \ref{sec acs}.
The complexity of (\ref{a57}) is about $O(M^2 )$, which is smaller
than $O(2M^2 )$ of (\ref{a51}). Moreover, if there is an efficient
algorithm to do the multiplication by ${\bf{Q}}$, the complexity can be further reduced. For example,
when ${\mathbf{Q}}$ is a discrete Fourier transform (DFT) matrix,
we can use a structure similar to FFT to obtain a more efficient
algorithm with complexity  $ O(M\log M)$\cite{intf}. The drawback
of this method is the accumulated rounding error. This leads to some signals with relatively
high PAPR. However, we can still expect that the shaped signals have low
PAPR values with high probability.

It is more convenient and better for shaping to use the
complex representation. Thus (\ref{a31}) becomes
\begin{equation}
{\mathbf{x}} = {\mathbf{Gs}},\;\;\;\;\;\;{\mathbf{s}} \in \left(
{\mathbb{Z}[i]} \right)^\frac{M}{2} ,{\mathbf{G}} \in \mathbb{C}^{\frac{M}{2} \times \frac{M}{2}}.
\end{equation}
When using the complex representation, the corresponding $j_m$ in
SERM and TERM can be $\pm 1$ or $\pm i$ and $ \left[
{\;} \right]$ denotes rounding the real and imaginary components
individually. The inverse operations (\ref{invu}), (\ref{invs}) still
work. There is a corresponding theorem \cite{matf} as follows
\\
\begin{thm}
Matrix ${\bf{Q}}$  has a factorization of $ {\mathbf{Q}} =
{\mathbf{PLD}}_{\mathbf{R}} {\mathbf{US}}_0$ if and only if $\det
({\mathbf{Q}}) = \det ({\mathbf{D}}_{\mathbf{R}} ) \ne 0$, where
${\mathbf{D}}_{\mathbf{R}}  = diag(1,1,...,1,e^{i\theta } )$,
${\mathbf{L}}, {\mathbf{U}}$ are  lower and  upper TERMs,
respectively, and ${\bf{P}}$ is a permutation matrix.
\label{thm6}\\
\end{thm}
If $\det ({\mathbf{Q}}) =  \pm 1$ or $\pm i$, we have a
simplified factorization,  $ {\mathbf{Q=PLUS}}_0$. It is in fact a
generalization of the lifting schemes in \cite{newn}.

When $\det ({\mathbf{Q}}) = e^{i\theta }$ is not equal to $\pm 1$ or
$\pm i$, a complex rotation $
e^{i\theta }$ can be implemented with the real and imaginary
components of a complex number and factorized into three unit TERMs as
\[
\begin{gathered}
  \left( {\begin{array}{*{20}c}
   {\cos \theta } & { - \sin \theta }  \\
   {\sin \theta } & {\cos \theta }  \\

 \end{array} } \right) = \left( {\begin{array}{*{20}c}
   1 & 0  \\
   {\left( {1 - \cos \theta  } \right)/\sin \theta  } & 1  \\

 \end{array} } \right)\left( {\begin{array}{*{20}c}
   1 & { - \sin \theta }  \\
   0 & 1  \\

 \end{array} } \right) \\
  \;\;\;\;\;\;\;\;\;\;\;\;\;\;\;\; \cdot \;\left( {\begin{array}{*{20}c}
   1 & 0  \\
   {\left( {1 - \cos \theta } \right)/\sin \theta } & 1  \\

 \end{array} } \right) \\
\;\;\;
\;\;\;\;\;\;\;\;\;\;\;\;\;\;\;\;\;\;\;\;\;\;\;\;\;\;\;\;\;\;\;\;\;
   = \left( {\begin{array}{*{20}c}
   1 & {\left( {\cos \theta  - 1} \right)/\sin \theta }  \\
   0 & 1  \\

 \end{array} } \right)\left( {\begin{array}{*{20}c}
   1 & 0  \\
   \sin \theta & 1  \\

 \end{array} } \right) \\
  \;\;\;\; \;\;\;\;\;\;\;\;\;\;\;\;\;\;\cdot \;\left( {\begin{array}{*{20}c}
   1 & {\left( {\cos \theta  - 1} \right)/ \sin \theta }  \\
   0 & 1  \\

 \end{array} } \right). \\
\end{gathered}
\]
 Therefore, Theorem \ref{thm6}
shows that given a nonsingular matrix, we can always derive an
integer reversible mapping by a factorization, which is  what we
need for constellation shaping.
\section{Simulation Results}\label{sec sr}
 In this section, we present simulation
results for shaping of space-time codes designed in \cite{pers}
and \cite{perb} by using $10^{6}$ randomly generated symbols.
Since the signals transmitted by the antennas have similar
statistical distributions, the simulation results are presented as
the average complementary cumulative density function (CCDF) of
the PAPR of signals on each antenna $i$, expressed as follows:
\begin{equation}
CCDF\{PAPR(x_{i})\}=P\{PAPR(x_{i})>\rho_{i}\}, \label{a4.1}
\end{equation}
where $PAPR(x_{i})\triangleq\frac{{|x_i |^2 }} {{E_{x_i } [|x_i |^2
]}}$. This can be interpreted as the probability that the PAPR of a
symbol $x_{i}$ exceeds a certain PAPR constraint, $\rho_{i}$.

We first look at the $4\times 4$ space-time code designed in
\cite{perb}, which achieves the D-MG tradeoff. Fig. \ref{fig4.1}
shows the CCDF of the PAPR on $4$ antennas using the HNF and PLUS
approximate cubic shaping introduced in Section \ref{sec HNF} and
Section \ref{sec PLUS}, respectively. The effect of the
constellation size is also investigated. When the constellation size
is moderate (64 QAM), it is observed that the HNF shaping method
results in about 1.3dB larger reduction in the PAPR than the PLUS
shaping, which provides about 2dB PAPR reduction. The PLUS shaping has a worse performance due to the accumulation of rounding errors (\ref{a59}).
As the constellation size becomes large (dense), we can expect that
the rounding error becomes relatively small, and both methods' PAPR
will approach the optimal value for cubic shaping, namely
$10log3=4.78$dB. This trend is shown by the curves of the PLUS
shaping. The HNF shaping result with 256 QAM was not obtained due to
its excessively high computational complexity. Table ~\ref{tb4.1} shows the
increased average power (compared to the average power without
shaping) due to the few points outside the hypercube. As the
constellation size becomes large (and more cubic), the power increment
decreases.
\begin{figure}[htbp]
\centering
\includegraphics[scale=1, width=0.5\textwidth]{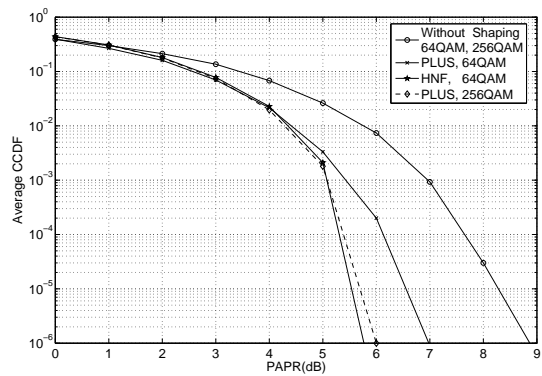}
\caption{CCDF of PAPR for a $4\times 4$ space-time code \cite{perb} using HNF and PLUS approximate cubic shaping.}
\label{fig4.1}
\end{figure}

\begin{table}
\centering
\begin{tabular}[htbp]{|c|c|c|}
\hline
 & HNF shaping & PLUS shaping \\
\hline
64QAM & 4.9\% & 4.6\% \\
256QAM &  & 3.5\% \\
 \hline
\end{tabular}
\caption{Increased average power for a $4\times 4$ space-time code \cite{perb}
using HNF and PLUS approximate cubic shaping.} \label{tb4.1}
\end{table}

In Fig. \ref{fig4.2}, we investigate the $5\times5$ space-time
code given in \cite{pers} which also achieves the D-MG tradeoff.
Similar trends as in the $4 \times 4$ case can be observed.
\begin{figure}[htbp]
\centering
\includegraphics[scale=1, width=0.5\textwidth]{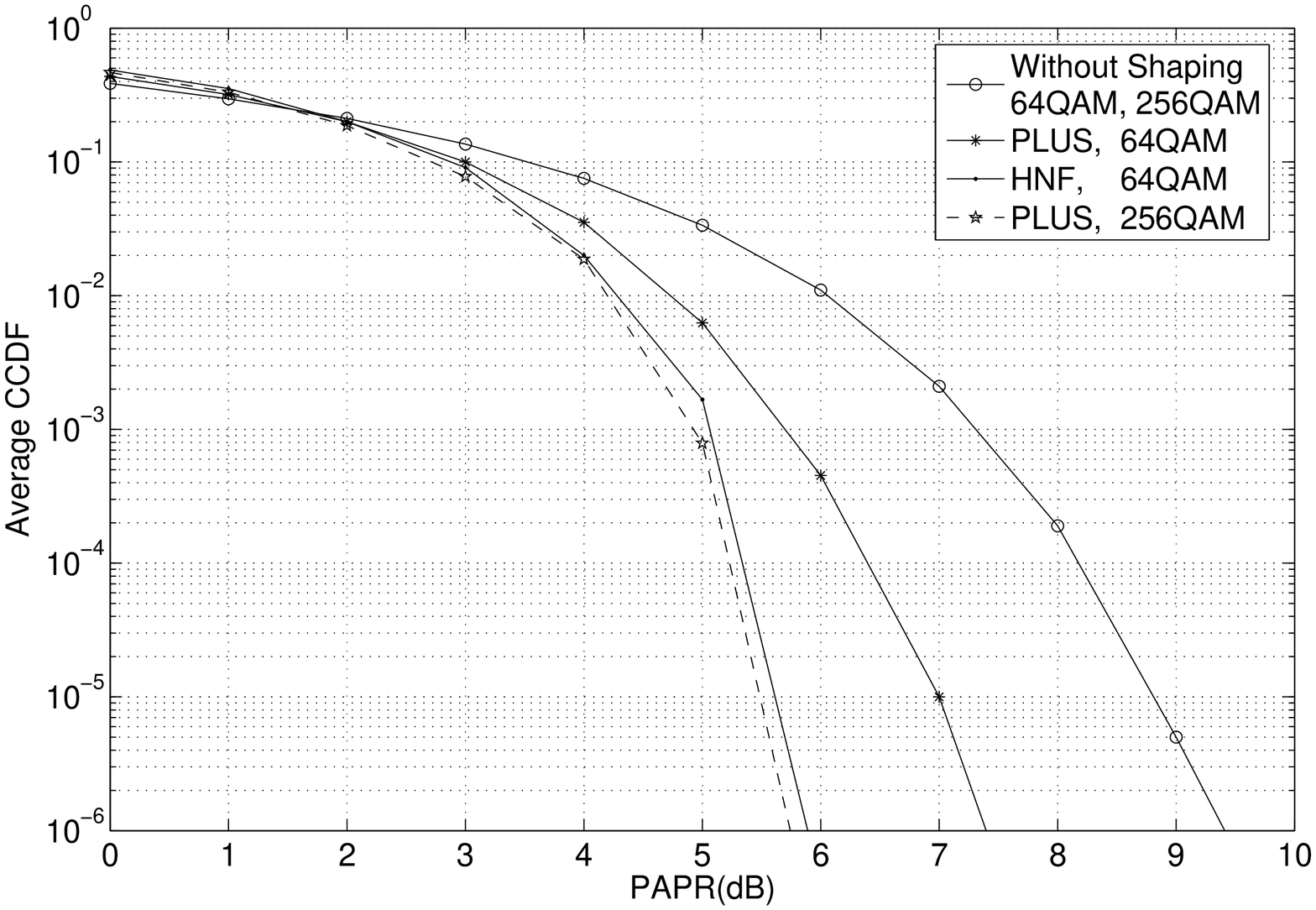}
\caption{CCDF of PAPR for a $5\times 5$ space-time code \cite{pers} using HNF
and PLUS approximate cubic shaping.} \label{fig4.2}
\end{figure}
\begin{table}
\centering
\begin{tabular}[htbp]{|c|c|c|}
\hline
 & HNF shaping & PLUS shaping \\
\hline
64QAM & 5.4\% & 4.8\% \\
256QAM & & 3.2\% \\
 \hline
\end{tabular}
\caption{Increased average power for a $5\times 5$ space-time code \cite{pers}
using HNF and PLUS approximate cubic shaping.} \label{tb4.2}
\end{table}

Finally, Fig. \ref{fig4.3} presents the codeword error probability (CEP)
of systems with $4$ or $5$ receive antennas and $4$ or $5$ transmit
antennas in quasi-static Rayleigh fading channels. Here, we use the
perfect space-time codes in \cite{pers} and \cite{perb} for the
$4\times4$ and $5\times5$ channels, respectively. The codeword sizes are
also $4\times4$ and $5\times5$ symbols, respectively. The sphere decoder in
\cite{aunified} is used for lattice decoding. The results show that
the space-time codes after shaping yield almost indistinguishable
error performance compared to the performance without shaping.

\begin{figure}[htbp]
\centering
\includegraphics[scale=1, width=0.5\textwidth]{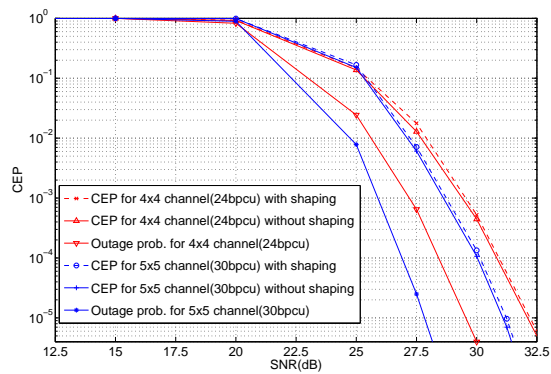}
\caption{Codeword error probability for Rayleigh fading channel with or without shaping.}
\label{fig4.3}
\end{figure}

\section{Conclusion}\label{sec conclusion}
In this paper, we first showed that, for Rayleigh fading channels, the D-MG tradeoff remains unchanged with {\it any} PAPR constraints larger than one. This result implies that, instead of designing codes on a case-by-case basis, as done by most existing works, there possibly exist general methodologies for designing space-time codes with low PAPR that achieve the optimal D-MG tradeoff. As an example of such methodologies, we proposed a PAPR reduction method based on constellation shaping that can be applied to existing optimal space-time codes without affecting
their optimality in the D-MG tradeoff. Unlike most PAPR reduction methods, the proposed method does not introduce redundancy or require side information being transmitted to the decoder.
Two realizations of the proposed method were considered. The first utilizes the Hermite Normal Form decomposition of integer matrices. The second utilizes the
integer reversible mapping.
Compared to the previous works
\cite{intc}\cite{pear} which applied a similar approach (Smith
Normal Form) to the single-antenna OFDM systems, the proposed
method has lower complexities. In addition, even though \cite{intc} managed to reduce the complexity to the same
order $O(MlogM)$ as the proposed integer reversible mapping scheme (in the single-antenna OFDM case)
by using a Hadamard matrix, that approach affects the PAPR reduction
capability and only works for OFDM systems. The proposed
method, on the other hand, works for any nonsingular generator
(modulation) matrix and can achieve better PAPR reduction. Sphere
decoding was performed to verify that the proposed PAPR reduction
method does not affect the optimality of space-time codes.
%

\appendices
\section{Proof of Lemma 1} \label{appxa}
Following the method in \cite{capo}, since the receiver knows the
realization of ${\bf{H}}$, the channel output is the pair $
({\mathbf{y, H}})$. The mutual information between channel input and
output is then
%
\begin{equation}
I\left( {{\mathbf{x;(y,H)}}} \right) = I{\mathbf{(x;H)}} +
I({\mathbf{x;y|H}}) = I({\mathbf{x;y|H}}). \label{a9}
\end{equation}
Denote $h(x)$ as the differential entropy of ${x}$ and let $H$ be a
particular realization of ${\bf H}$. 
For this ${H}$, when the SNR is asymptotically large, the output
differential entropy $h({\mathbf{y|H}} =
 H)$ can be well approximated by the input differential entropy $h({\mathbf{x|H}} =
 H)$. In addition,
 \begin{align}
 I({\mathbf{x;y|H}} = H) &= h({\mathbf{x|H}} = H) - h({\mathbf{x|y}},{\mathbf{H}} =
 H)\label{a10}\\
 &= h({\mathbf{x|H}} = H) - h({\mathbf{e}}|{\mathbf{y}},{\mathbf{H}} =H),
 \label{a11}
\end{align}
where $\;{\mathbf{e}} \triangleq {\mathbf{x}} - {\mathbf{F}}_{MMSE}
{\mathbf{y}}$, and ${\mathbf{F}}_{MMSE} $ is the minimum mean-square
error (MMSE) estimation filter of ${\bf{x}}$ given ${\bf{y}}$.

Since the lemma is to lower bound the ergodic channel capacity, any rate achieved by a particular signal can serve as a lower bound. We select the transmitted signal $\mathbf{x}$ such that $E\left[\mathbf{x}\right]=0$ and ${\mathbf{S}}_{{\mathbf{xx}}} \triangleq E\left[ {{\mathbf{xx}}^\dag  } \right]$ is positive definite.\footnote{Since space-time codes are open-loop solutions for which the transmitter does not have the channel state information, with identical complex Gaussian distributions of the fading coefficients among antennas (as assumed in Section \ref{systemdef}), a reasonable selection is to distribute the transmission power evenly on all the transmit antennas, and let $E\left[\mathbf{x}\right]=0$ for power efficiency. Together with additional selections, for example, simply letting the entries of $\mathbf{x}$ be independent of one another, ${\mathbf{S}}_{{\mathbf{xx}}}$ becomes positive definite (when the average total transmission power is not zero).}
In the following, we will compute the rate achievable by signals with these properties. This achievable rate obviously lower bounds the capacity.

Denote ${\mathbf{S}}_{{\mathbf{xy}}} \triangleq E\left[ {{\mathbf{xy}}^\dag
} \right]$.
According to the Orthogonality Principle, we have
\begin{align}
{\mathbf{F}}_{MMSE} \; &= {\mathbf{S}}_{{\mathbf{xy}}}
{\mathbf{S}}_{{\mathbf{yy}}}^{{\mathbf{ - 1}}} \;\notag\\
 &= {\mathbf{S}}_{{\mathbf{xx}}} H^\dag
{\mathbf{(}}H{\mathbf{S}}_{{\mathbf{xx}}} H^\dag  {\mathbf{ +
I)}}^{{\mathbf{ - 1}}} \notag\\
&= {\mathbf{(}}H^\dag  H{\mathbf{ + S}}_{{\mathbf{xx}}}^{{\mathbf{ -
1}}} {\mathbf{)}}^{{\mathbf{ - 1}}} H^\dag \label{a12}
\end{align}
where
the matrix inversion lemma
\[{\mathbf{(A +
BCD)}}^{{\mathbf{ - 1}}} {\mathbf{ = A}}^{{\mathbf{ - 1}}} {\mathbf{
- A}}^{{\mathbf{ - 1}}} {\mathbf{B(C}}^{{\mathbf{ - 1}}} {\mathbf{ +
DA}}^{{\mathbf{ - 1}}} {\mathbf{B)}}^{{\mathbf{ - 1}}}
{\mathbf{DA}}^{{\mathbf{ - 1}}}
\]
is used. $E\left[ {{\mathbf{ee}}^\dag  } \right]$ can be computed as
\begin{align}
E\left[ {{\mathbf{ee}}^\dag  } \right] &=
{\mathbf{S}}_{{\mathbf{xx}}} {\mathbf{ - S}}_{{\mathbf{xx}}} H^\dag
{\mathbf{(}}H{\mathbf{S}}_{{\mathbf{xx}}} H^\dag  {\mathbf{
+ I)}}^{{\mathbf{ - 1}}} H{\mathbf{S}}_{{\mathbf{xx}}} \notag\\
&= {\mathbf{(}}H^\dag  H{\mathbf{ + S}}_{{\mathbf{xx}}}^{{\mathbf{ -
1}}} {\mathbf{)}}^{{\mathbf{ - 1}}} \label{a13}
\end{align}
where the matrix inversion lemma is again used.
Note that $E[{\mathbf{e}}]=0$ since $E[{\mathbf{x}}]=0$ and $E[{\mathbf{y}}]=0$. Thus the covariance matrix of ${\bf{e}}$,
denoted $ {\mathbf{Cov}}[{\mathbf{e}}]$,
is equal to $E[{\mathbf{ee}}^\dag  ]$. Then we have

%
%
%
%
\begin{equation}
\begin{split}
h({\mathbf{e}}|{\mathbf{y}},{\mathbf{H}} = H) \leq
h({\mathbf{e}}|{\mathbf{H}} = H) & \leq \log \det \left( {\pi
e{\mathbf{Cov}}[{\mathbf{e}}]} \right)
 \\ & = \log \det \left( {\pi eE[{\mathbf{ee}}^\dag  ]}
 \right)\;
 \label{a15}
 \end{split}.
 \end{equation}
Define 
\begin{align}
\tilde I({\mathbf{x;y|H}} = H) &\triangleq h({\mathbf{x|H}} = H) -
\log \det \left( {\pi eE[{\mathbf{ee}}^\dag  ]}
 \right)\notag\\
 &= h({\mathbf{x|H}} = H) + \log \det \left( {\frac{1}
{{\pi e}}{\mathbf{(}}H^\dag  H{\mathbf{ +
S}}_{{\mathbf{xx}}}^{{\mathbf{ - 1}}} {\mathbf{)}}} \right).
\label{a16}
\end{align}
Obviously, $\tilde I({\mathbf{x;y|H}} = H) \leq I({\mathbf{x;y|H}} =
H)$. We have the ergodic capacity
\[\begin{gathered}
  C = \mathop {\max }\limits_{f_{\mathbf{x}} ({\mathbf{x}})} I\left( {{\mathbf{x;(y,H)}}} \right)\mathop  = \limits^{(\ref{a9})} \mathop {\max }\limits_{f_{\mathbf{x}} ({\mathbf{x}})} \;E_{\mathbf{H}} \left[ {I({\mathbf{x;y|H}})} \right] \hfill \\
  \;\;\;\mathop  = \limits^{(\ref{a11})} \mathop {\max }\limits_{f_{\mathbf{x}} ({\mathbf{x}})} \;E_{\mathbf{H}} \left[ {h({\mathbf{x|H}}) - h({\mathbf{e}}|{\mathbf{y}},{\mathbf{H}})} \right] \hfill \\
\end{gathered} \]
where $f_{\mathbf{x}} ({\mathbf{x}})$ is the probability density
function of ${\bf{x}}$ subject to (\ref{a7}) and
(\ref{a8}). The ergodic capacity is lower bounded by
\begin{align}
C &= \mathop {\max }\limits_{f_{\mathbf{x}} ({\mathbf{x}})} \;E_{\mathbf{H}} \left[ {I({\mathbf{x;y|H}})} \right]
 \geq  \mathop {\max }\limits_{f_{\mathbf{x}} ({\mathbf{x}})} \;E_{\mathbf{H}} \left[ {\tilde I  ({\mathbf{x;y|H}})} \right]\notag\\
  &= \mathop {\max }\limits_{f_{\mathbf{x}} ({\mathbf{x}})} \;E_{\mathbf{H}} \left[ {h({\mathbf{x|H}}) + \log \det \left( {\frac{1}
 {{\pi e}}{\mathbf{(H}}^\dag  {\mathbf{H + S}}_{{\mathbf{xx}}}^{{\mathbf{ - 1}}} {\mathbf{)}}} \right)\;}
 \right]\label{a17}\\
 &\geq \mathop {\max }\limits_{f_{\mathbf{x}} ({\mathbf{x}})} E_{\mathbf{H}} \left[ {h({\mathbf{x|H}})} \right] + \left(E_{\mathbf{H}} \left[ {\log \det \left( {\frac{1}
 {{\pi e}}{\mathbf{(H}}^\dag  {\mathbf{H + S}}_{{\mathbf{xx}}}^{{\mathbf{ - 1}}} {\mathbf{)}}} \right)\;}
 \right]\right)_{f_{\mathbf{x}} ({\mathbf{x}})=f^*_{\mathbf{x}} ({\mathbf{x}})}\label{a18}\\
  &= \mathop {\max }\limits_{f_{\mathbf{x}} ({\mathbf{x}})} \;h({\mathbf{x}}) + \left(E_{\mathbf{H}} \left[ {\log \det \left( {\frac{1}
 {{\pi e}}{\mathbf{(H}}^\dag  {\mathbf{H + S}}_{{\mathbf{xx}}}^{{\mathbf{ - 1}}} {\mathbf{)}}} \right)}
 \right]\right)_{f_{\mathbf{x}} ({\mathbf{x}})=f^*_{\mathbf{x}} ({\mathbf{x}})}\label{a19}\\
  &\triangleq C',\notag
  \end{align}
  %
%
where $f^*_{\mathbf{x}} ({\mathbf{x}})=\arg \mathop {\max }\limits_{f_{\mathbf{x}} ({\mathbf{x}})} E_{\mathbf{H}} \left[ {h({\mathbf{x|H}})} \right] = \arg \mathop {\max }\limits_{f_{\mathbf{x}} ({\mathbf{x}})} \;h({\mathbf{x}})$. $f^*_{\mathbf{x}} ({\mathbf{x}})$ and $C'$ can be obtained by solving the following problem
\begin{equation}
\begin{gathered}
\mathop {\max }\limits_{f_{\mathbf{x}} ({\mathbf{x}})}
\;h({\mathbf{x}})\\
\mbox{s.t.}\;\;\;\mbox{Tr}\left( {E_{\mathbf{x}} \left[
{{\mathbf{xx}}^\dag } \right]} \right) \leq
{\mathbf{P}}\\
\;\;\;\; \;\;\;\;\;\;\;\;\;\;\;\;\;\;\; \frac{{|x_i |^2 }} {{E_{x_i
} [|x_i |^2 ]}} \leq \;\rho _i, \;\;\; i=1,\ldots,m. \label{a20}
\end{gathered}
\end{equation}
Due to the circular symmetry of the constraints (\ref{a7}) and
(\ref{a8}), polar coordinates
\[\begin{gathered}
  {\mathbf{x}} = [x_1 ,x_2 ...,x_m ]^T  = [r_1 e^{j\theta _1 } ,r_2 e^{j\theta _2 } ,...,r_m e^{j\theta _m } ]^T  \hfill \\
  \;\;\;\;\;\;\;\;\;\;\;\;\;\;\;\;\;r_i  \geq 0,\;\;\;\theta _i  \in [0,2\pi ) \hfill \\
\end{gathered} \]
are found convenient, where $r_i $ and $\theta _i$ stand,
respectively, for the amplitude and phase of
$x_i$. Straightforward transformation yields
\[\begin{split}
h({\mathbf{x}}) &=  - \int {f_{\mathbf{x}} ({\mathbf{x}})} \log
f_{\mathbf{x}} ({\mathbf{x}})d{\mathbf{x}} =  - \int
{f_{{\mathbf{r,\theta }}} {\mathbf{(r,\theta )}}} \log
\frac{{f_{{\mathbf{r,\theta }}} {\mathbf{(r,\theta )}}}}
{{\prod\limits_{i = 1}^m {r_i } }}d{\mathbf{r}}d{\mathbf{\theta }}\\
  &= h({\mathbf{r,\theta }}) + \sum\limits_{i = 1}^m {\left( {\int {f_{r_i } {\mathbf{(}}r_i {\mathbf{)}}} \log r_i dr_i } \right)}
 \end{split}\]
where ${\mathbf r}$ and ${\mathbf \theta}$ are vectors consisting of $r_i $ and $\theta _i$, respectively. Note that
 \[
 h({\mathbf{r,\theta }}) \leq h({\mathbf{r}}) + h({\mathbf{\theta }}) \leq h({\mathbf{r}}) + m\log \;2\pi
 .\]
Therefore, to maximize $h({\mathbf{x}})$, we should choose
${\mathbf{r}}$ and ${\mathbf{\theta }}$ independent of each other, and all
$\theta_i$ distributed independently and uniformly in $[0,2\pi )$. Then
the equality holds and
\[
h({\mathbf{x}}) = h({\mathbf{r}}) + \sum\limits_{i = 1}^m {\left(
{\int {f_{r_i } {\mathbf{(}}r_i {\mathbf{)}}} \log r_i dr_i }
\right)}  + m\log 2\pi .
\]
Similarly,
\[
h({\mathbf{r}}) \leq \sum\limits_{i = 1}^m {h(r_i )}.
\]
Choosing $r_i $ independent of one another, the equality holds
and $h({\mathbf{x}})$ is maximized.\footnote{Note that the selection of independent $\theta_i$'s and $r_i $'s is one of the possible selections we made in the previous footnote to make ${\mathbf{S}}_{{\mathbf{xx}}}$ positive definite.} Drop the last term of
$h({\mathbf{x}})$, and transform (\ref{a20}) into the following
equivalent optimization problem
\begin{equation}
\begin{gathered}
  \;\;\;\mathop {\max }\limits_{f_{\mathbf{r}} {\mathbf{(r)}}} \;\left( {\sum\limits_{i = 1}^m {h(r_i )}  + \sum\limits_{i = 1}^m {\left( {\int {f_{r_i } {\mathbf{(}}r_i {\mathbf{)}}} \log r_i dr_i } \right)} } \right) \hfill \\
   = \mathop {\max }\limits_{f_{\mathbf{r}} {\mathbf{(r)}}} \;\left( { - \sum\limits_{i = 1}^m {\left( {\int {f_{r_i } {\mathbf{(}}r_i {\mathbf{)}}} \log \frac{{f_{r_i } {\mathbf{(}}r_i {\mathbf{)}}}}
{{r_i }}dr_i } \right)} } \right)\; \hfill \\
  \;\;\;\;\;\;\;\;\;\;\;\;\;\;\;\;\;\;\mbox{s.t.}\;\;\;\mbox{Tr}\left( {E_{\mathbf{r}} \left[ {{\mathbf{rr}}^{\mathbf{*}} } \right]} \right) \leq {\mathbf{P}} \hfill \\
  \;\;\;\;\;\;\;\;\;\;\;\;\;\;\;\;\;\;\;\;\;\;\;\;\;\;\;\;\;\;\;\frac{{|r_i |^2 }}
{{E_{r_i } [|r_i |^2 ]}} \leq \;\rho_{i}, \;\;\; i=1,\ldots,m.  \hfill \label{a21}\\
\end{gathered}
\end{equation}
 For each antenna $i$, given the transmission power $P_i$ such that
\[
\sum\limits_{i = 1}^m {P_i }  = {\mathbf{P}}
\]
and a PAPR constraint $\rho _i$, similar to \cite{thec}, the optimal
solution $f_{r_i }^* (r_i )$ is (see Appendix \ref{appxB})
\begin{equation}
\begin{split}
f_{r_i }^* (r_i ) &= a_i r_i \;\exp ( - b_i r_i ^2
/2),\;\;\;\;\;\forall r_i  \in \left[ {0,\sqrt {\rho _i P_i } }
\right]\\
f_{r_i }^* (r_i ) &=
0,\;\;\;\;\;\;\;\;\;\;\;\;\;\;\;\;\;\;\;\;\;\;\;\;\;\;\;\;\;\;\;\;\;\;\;\;\forall
r_i \notin \left[ {0,\sqrt {\rho _i P_i } } \right] \label{a22}
\end{split}
\end{equation}
where $a_i$, $b_i$ satisfy (\ref{a23}), (\ref{a24}) or (\ref{a25}):
\begin{align}
&\mbox{when}\;\;\rho _i  \ne 2,\; \rho_i>1 \notag\\
&\frac{{a_i }}
{{b_i }}(1 - \exp ( - b_i \rho _i P_i /2)) = 1\label{a23}\\
&2(a_i /b_i )(b_i \rho _i P_i )^{ - 1} [1 - (1 + b_i \rho _i P_i
/2)\exp ( - b_i \rho _i P_i /2)] = 1/\rho
_i\label{a24},\\
\notag\\
&\mbox{when}{\text{  }}\rho _i  = 2\notag\\
&a_i  = \frac{2} {{\rho _i P_i \;}},    \;\;\;b_i  = 0.\label{a25}
\end{align}
Denoting the maximum of $h(x_i )$ as $h^* (x_i )$ and $c_i  = b_i
\rho _i P_i$, we can compute $h^* (x_i )$ directly by using $f_{r_i
}^* (r_i )$
\begin{align}
h^* (x_i  ) &=  - \log a_i  + \frac{{b_i P_i }}
{2} + \log 2\pi \label{a60}\\
& = \; - \log a_i  + \frac{{c_i }}
{{2\rho _i }} + \log 2\pi\\
&=\left.
\begin{cases}
  \log P_i  + \log \frac{{\rho _i (1 - \exp ( -
c_i /2))}} {{c_i }}
+ \frac{{c_i }} {{2\rho _i }} + \log 2\pi,  \;\;\;\;\;\;  &\rho_i\neq2, ~\rho_i >1\\
\log {{2\pi}{P_i}},  &\rho_i=2
\end{cases}
\right..
\end{align}
From (\ref{a23}), $(1 - \exp ( - b_i
\rho _i P_i /2))^{-1} = \frac{{a_i }} {{b_i }}$. By substituting $\frac{{a_i }} {{b_i }}$ in (\ref{a24}) with $(1 - \exp ( - b_i
\rho _i P_i /2))^{-1}$ and then replacing $b_i \rho _i P_i$ with $c_i$, we will arrive at
\[
\frac{2}{c_i}-\frac{1}{1 - \exp ( - c_i/2)}+1=\frac{1}{\rho_i}
\]
which indicates that $1/\rho_i$ is a monotonic function of $c_i$, as shown in Fig.~\ref{figA1}. Thus,
when $\rho _i > 1$ is fixed and finite, $c_i$ is a finite constant.
With the independence between $x_i$'s,
\[
h^* ({\mathbf{x}}) = \sum\limits_{i = 1}^m {h^* (x_i )}.
\]
Now we can plug $h^* ({\mathbf{x}})$ and the corresponding (independent) distribution of ${\mathbf{x}}$ into (\ref{a19}) to obtain the lower bound $C'$ of the ergodic capacity.
Let
%
\[
k_i  =
\begin{cases}
 \log \frac{{\rho _i (1 - \exp ( - c_i /2))}} {{c_i }} +
\frac{{c_i }} {{2\rho _i }} + \log \frac{2} {e},   \;\;\;\;\;\;&\rho_i\neq2, ~\rho_i>1\\
 \log \frac{2} {e}, &\rho_i=2
\end{cases}
\]
which is a constant because $c_i$ is a finite constant when $\rho _i > 1$ is fixed and finite.
%
We have
\[
C' = E_{\mathbf{H}} \left[ {\log \det \left( {{\mathbf{I}} +
{\mathbf{H}}  {\mathbf{S}}_{{\mathbf{xx}}} {\mathbf{H}}^\dag}
\right)} \right] + \sum\limits_{i = 1}^m {k_i },
\]
where the equality follows from the determinant identity $\det ({\mathbf{I
+ AB}}) = \det ({\mathbf{I + BA}})$.
With the selection of equal-power allocation, $P_i = {\mathbf{P}}/m$, $\forall i$. Thus
\begin{equation}
C \geq C'  = E_{\mathbf{H}} \left[ {\log \det \left( {{\mathbf{I}}
+ \frac{{\mathbf{P}}} {m}{\mathbf{H}} {\mathbf{H}}^\dag} \right)}
\right] + \sum\limits_{i = 1}^m {k_i }.\label{a29}
\end{equation}
%
Note that the inequality holds for any distribution of $\mathbf{H}$.
 When $\rho _i \to \infty$, $a_i = b_i = 2/P_i$, and
\[
C=C'   = E_{\mathbf{H}} \left[ {\log \det \left( {{\mathbf{I}} +
\frac{{\mathbf{P}}} {m}{\mathbf{H}}  {\mathbf{H}}^\dag} \right)}
\right]
\]
which is the classical result without PAPR constraints. The constant $k_i$ (i.e., the difference between $h^* (x_i)$ when $\rho_i \to \infty$ and when $\rho_i$ is finite) is shown in Fig.~\ref{figA2}.
\section{Derivation of the Optimal Signal Probability Density Functions} \label{appxB}
We consider the following optimization problem with average power and PAPR constraints
\begin{equation}
\begin{gathered}
  \;\;\mathop {\max }\limits_{f_r {\mathbf{(}}r{\mathbf{)}}} \;\left( { - \int {f_r {\mathbf{(}}r{\mathbf{)}}} \log \frac{{f_r {\mathbf{(}}r{\mathbf{)}}}}
{r}dr} \right)\; \hfill \\
  \;\;\;\;\mbox{s.t.}\;\;\;\;\;E\left[ {|r|^2 } \right] \leq P \hfill \\
  \;\;\;\;\;\;\;\;\;\;\;\;\;\;\frac{{|r|^2 }}
{{E[|r|^2 ]}} \leq \;\rho. \;\;\;\; \hfill \\
\end{gathered}
\label{1a}
\end{equation}
Note that the PAPR constraint is different from the {\it peak} power constraint considered in \cite{thec}, thus the results in \cite{thec} do not directly apply to our case. The following derivation (verification) is necessary.
The problem will be solved through the
following slightly different problem with average and peak
power constraints
\begin{equation}
\begin{gathered}
  \;\;\;\mathop {\max }\limits_{f_r {\mathbf{(}}r{\mathbf{)}}} \;\left( { - \int {f_r {\mathbf{(}}r{\mathbf{)}}} \log \frac{{f_r {\mathbf{(}}r{\mathbf{)}}}}
{r}dr} \right)\; \hfill \\
  \;\;\;\;\mbox{s.t.}\;\;\;\;\;E\left[ {|r|^2 } \right] \leq P \hfill \\
  \;\;\;\;\;\;\;\;\;\;\;\;\;\;\;\;\;\;\;\;\;|r|^2  \leq \;\rho P\; \hfill \\
\end{gathered}
\label{2a}
\end{equation}
which can be rewritten as
\begin{equation}
\begin{gathered}
  \mathop {\max }\limits_{f_r {\mathbf{(}}r{\mathbf{)}}}  - \left( {\int_0^{\sqrt {\rho P} } {f_r (r)\log \frac{{f_r (r)}}
{r}dr} } \right) \hfill \\
  \;\mbox{s.t.}\;\;\;\;\;f_r (r) \geq 0,\;\;\;\;\forall r \in [0,\sqrt {\rho P} ] \hfill \\
  \;\;\;\;\;\;\;\;\;\;\int_0^{\sqrt {\rho P} } {f_r (r)dr = 1}  \hfill \\
  \;\;\;\;\;\;\;\;\;\int_0^{\sqrt {\rho P} } {r^2 f_r (r)dr \leq P}. \;\; \hfill \\
\end{gathered}
\label{3a}
\end{equation}
The optimal solution $f_r ^* (r)$ of (\ref{3a}) is given by the
standard variational techniques \cite{thec}
\begin{align}
f_r ^* (r) &=ar{}\exp ( - br^2 /2),\;\;\;\;\;\;\;\;\ \forall r \in
\left[ {0,\sqrt {\rho P} } \right]\label{4a}\\
f_r ^* (r) &=\;\;\;\;
0,\;\;\;\;\;\;\;\;\;\;\;\;\;\;\;\;\;\;\;\;\;\;\;\;\;\;\;\;\;\;\;\;\forall
r \ne \left[ {0,\sqrt {\rho P} } \right]. \label{5a}
\end{align}
Observe that if the first equality in (\ref{2a}) holds, the optimal
solution $f_r ^* (r)$ of (\ref{3a}) is also the optimal solution of
(\ref{1a}). However, the equality does not always hold.

We discuss $a$, $b$ for different values of PAPR ($\rho >1$). When
$\rho>2$, $a$, $b$ satisfy
\begin{align}
&\frac{a} {b}(1 - \exp ( - b\rho P/2)) = 1\label{6a}\\
&2(a/b)(b\rho P)^{ - 1} [1 - (1 + b\rho P/2)\exp ( - b\rho P/2)] =
1/\rho. \;\label{7a}
\end{align}
Equations (\ref{6a}) and (\ref{7a}) together solve $b$ as a function of $\rho$
and $P$ which is illustrated in Fig.~\ref{figA1} with $T = \rho P/2$.
Fig.~\ref{figA1} shows that $b>0$ when $\rho>2$. Thus the first
equality in (\ref{2a}) holds and $f_r ^* (r)$ is also the optimal
solution of (\ref{1a}). Note that when $\rho \to \infty$, $ a = b = 2/P$
and $f_r ^* (r)$ is the Rayleigh distribution as expected.
\begin{figure}[htb]
\centering
\includegraphics[width=0.5\textwidth]{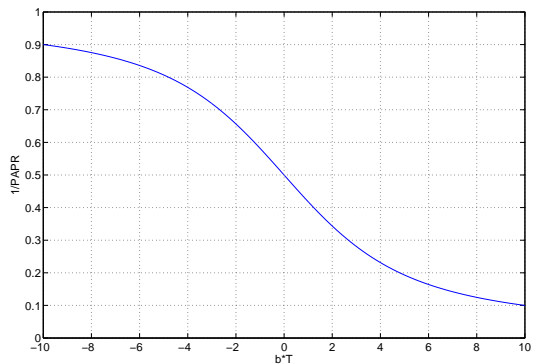}
\caption{The relation between $bT$ ($bT=b\rho P/2=c/2$, as defined in Appendix \ref{appxa}, before (\ref{a60})) and $1/\rho$ subject to
(\ref{6a}) and (\ref{7a}).} \label{figA1}
\end{figure}

When $\rho=2$, $a$, $b$ satisfy
\begin{align}
&b = 0,\;\;a = 1/P. \label{8a}
\end{align}
In this case, $f_r ^* (r)$ is linear and the equality in (\ref{2a})
is again satisfied and $f_r ^*(r)$ is the optimal solution of
(\ref{1a}).

For the case of $\rho<2$, the Karush-Kuhn-Tucker (KKT) conditions for the optimization problem (\ref{3a}) require that $b\geq0$.
However, from Fig.~\ref{figA1}, $b<0$ when $\rho<2$. Therefore, $a$,
$b$ for the optimal solution of (\ref{3a}) should satisfy (\ref{8a}).
In this situation
\begin{equation}
\int_0^{\sqrt {\rho P} } {r^2 f_r ^* (r)dr = \frac{\rho } {2}P < P}.
\label{9a}
\end{equation}
That is, the first equality in (\ref{2a}) does not hold, and the corresponding PAPR value is $2$, larger than $\rho$. As a result, $f_r ^* (r)$ is not the optimal solution of (\ref{1a}).

To obtain the optimal solution of (\ref{1a}) when $\rho<2$, consider
the following problem with slightly different constraints
\begin{equation}
\begin{gathered}
  \;\;\mathop {\max }\limits_{f_r {\mathbf{(}}r{\mathbf{)}}} \;\left( { - \int {f_r {\mathbf{(}}r{\mathbf{)}}} \log \frac{{f_r {\mathbf{(}}r{\mathbf{)}}}}
{r}dr} \right)\; \hfill \\
  \;\;\;\;\mbox{s.t.}\;\;\;\;\;E\left[ {|r|^2 } \right] = P \hfill \\
  \;\;\;\;\;\;\;\;\;\;\;\;\;\;\;\;\;\;\;\;\;|r|^2  \leq \;\rho P.\; \hfill \\
\end{gathered}
\label{10a}
\end{equation}
Using similar optimization techniques, the optimal solution of (\ref{10a}), $f_r' (r)$, is found to have the
same form as (\ref{4a}), (\ref{5a}) with $b<0$. Therefore, $f_r'
(r)$ is not a Rayleigh-like distribution. We will show that $f_r'
(r)$ is also the optimal solution of (\ref{1a}). Assuming that the
distribution $ f_r'' (r)$ is the optimal solution of (\ref{1a}) and
$ f_r'' (r) \neq f_r' (r)$, then $f_r'' (r)$ must be the optimal
solution of the following optimization problem for some $P''$
\begin{equation}
\begin{gathered}
  \;\mathop {\max }\limits_{f_r {\mathbf{(}}r{\mathbf{)}}} \;\left( { - \int {f_r {\mathbf{(}}r{\mathbf{)}}} \log \frac{{f_r {\mathbf{(}}r{\mathbf{)}}}}
{r}dr} \right)\; \hfill \\
  \;\;\;\;\mbox{s.t.}\;\;\;\;\;E\left[ {|r|^2 } \right] = P'' < P \hfill \\
  \;\;\;\;\;\;\;\;\;\;\;\;\;\;\;\;\;|r|^2  \leq \;\rho P'' < \rho P.\;\; \hfill \\
\end{gathered}
\label{11a}
\end{equation}
However, (\ref{10a}) has a larger maximum value, namely $( - \log a
+ bP/2)$, than that of (\ref{11a}) because $P>P''$, which implies
that $f_r' (r)$ maximizes (\ref{1a}). In Fig.~\ref{figA2}, we
demonstrate the maximum values of $h^* (x)$, where {$h^* (x) = ( -
\log a + bP/2 + \log 2\pi )$, for $\rho=
5,\;\;2,\;\;1.1,\;\;\infty$.
\begin{figure}[htb]
\centering
\includegraphics[width=0.5\textwidth]{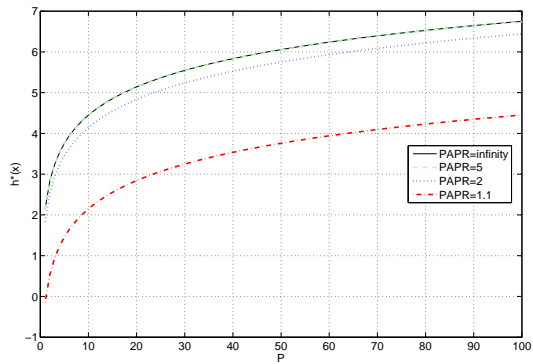}
\caption{The maximum value $h^* (x)$ for different $\rho$ (PAPR)
values}
 \label{figA2}
\end{figure}
\section{Proof of Theorem \ref{thm2}}\label{appxc}
We follow the method in \cite{diva}, letting the $i$-th element of
the input signal be drawn from the random code with {\it i.i.d.} distribution
$f_{x_i }^* (x_i )$
\begin{equation}
\begin{split}
f_{x_i }^* (x_i )&=\frac{1}{2\pi}
a_{i}exp\left(-\frac{b_{i}}{2}|x_{i}|^{2}\right),\;\;\;\; x_{i}\in
\textbf{B}_{i} \label{1b}\\
f_{x_i }^* (x_i
)&=0,\;\;\;\;\;\;\;\;\;\;\;\;\;\;\;\;\;\;\;\;\;\;\;\;\;\;\;\;\;\;\;\,\;\;\;\;\;\;\;\;\;x_{i}
\notin \textbf{B}_{i}
\end{split}
\end{equation}
where $x_{i}\in \mathbb{C}$, $ \textbf{B}_{i}\triangleq \{
x_{i}\mathbf{\mid} \;\;|x_{i}| \leq \sqrt{\rho_{i}P_{i}} \; \} $. At
data rate $R=rlogSNR$, the error probability is
\[
P_{e}(SNR)\leq P_{out}(R)+P(error,no\;outage).
\]
The second term can be upper bounded via a union bound. Assume that
$X(0)$, $X(1)$ are two possible transmitted codewords, and ${\Delta
X}={X(1)-X(0)}$. Suppose that $X(0)$ is transmitted. The probability that
a maximum likelihood receiver will make a detection error in favor
of $X(1)$, conditioned on a certain realization of the channel, is
\begin{align}
P\left( {X(0) \to X(1)|{\mathbf{H}} = H} \right) &= P\left( {\left\|
{\frac{1} {2}H {\Delta X}} \right\|^2 \leq
\left\| {\mathbf{w}} \right\|^2 } \right)\label{2b}\\
&\leq exp\left[-\frac{1}{4}\|H\Delta
X\|^{2}\right]\label{3b}
\end{align}
where $\mathbf{w}$ is the additive  noise on the direction of
$H\Delta X$. Then we need to average over the ensemble of random
codes. Let $x_i$ and $x'_i$ be two {\it i.i.d.} random variables with distribution in the form of (\ref{1b}), and $x'_i-x_i=\hat{x}_i$. The probability density function of $\hat{x}_i$ is
\begin{equation}
\begin{split}
f_{\hat{x_{i}}}(\hat{x_{i}})=\frac{1}{2\pi}a_{i}{}exp\left(-\frac{b_{i}}{4}|\hat{x_{i}}|^{2}\right)
\int_{x_{i}\in\mathbf{C}_{i}}\frac{1}{2\pi}a_{i}{}exp\left(-b_{i}|x_{i}-\frac{\hat{x_{i}}}{2}|^{2}\right)dx_{i}
\label{4b}
\end{split}
\end{equation}
where $x_{i}\in\mathbf{C}_{i}$ if $\;\;x_{i}\in\textbf{B}_{i}$ and
$x'_{i}\in\textbf{B}_{i}$. We discuss different values of $b_{i}$.

For $b_{i}>0$,
\begin{equation}
\int_{x_{i}\in\mathbf{C}_{i}}\frac{1}{2\pi}a_{i}{}exp\left(-b_{i}|x_{i}-\frac{\hat{x_{i}}}{2}|^{2}\right)dx_{i}\leq
\int_{x_{i}\in\Bbb{C}}\frac{1}{2\pi}a_{i}{}exp\left(-b_{i}|x_{i}-\frac{\hat{x_{i}}}{2}|^{2}\right)dx_{i}=t_{1}
\label{4b1}
\end{equation} where $t_{1}$ is a constant, which is independent of of $P_{i}$.

For $b_{i}=0$,
since $|x_{i}-\frac{\hat{x_{i}}}{2}|\leq2\sqrt{\rho_iP_i}$ and
$a_{i}P_i{\rho_i}=\rho_i$
\begin{equation}
\int_{x_{i}\in\mathbf{C}_{i}}\frac{1}{2\pi}a_{i}{}exp\left(-b_{i}|x_{i}-\frac{\hat{x_{i}}}{2}|^{2}\right)dx_{i}\leq
\int_{x_{i}\in\Bbb{C}}\frac{1}{2\pi}a_{i}{}exp\left(-a_{i}|x_{i}-\frac{\hat{x_{i}}}{2}|^{2}\right)exp\left(4a_{i}\rho_i{P_i}\right)dx_{i}=t_{2}
\label{4b2}
\end{equation}
where $t_{2}$ is a constant, which is independent of of $P_{i}$.

For $b_{i}<0$,
since $|x_{i}-\frac{\hat{x_{i}}}{2}|\leq2\sqrt{\rho_i{P_i}}$ and
$b_{i}P_i{\rho_i}$ is a constant.
\begin{equation}
\int_{x_{i}\in\mathbf{C}_{i}}\frac{1}{2\pi}a_{i}{}exp\left(-b_{i}|x_{i}-\frac{\hat{x_{i}}}{2}|^{2}\right)dx_{i}\leq
\int_{x_{i}\in\Bbb{C}}\frac{1}{2\pi}a_{i}{}exp\left(b_{i}|x_{i}-\frac{\hat{x_{i}}}{2}|^{2}\right)exp\left(-8b_{i}\rho_i{P_i}\right)dx_{i}=t_{3}
\label{4b3}
\end{equation}
where $t_{3}$ is a constant, which is independent of of $P_{i}$.

Thus  we have
\begin{align}
f_{\hat{x_{i}}}(\hat{x_{i}}) &\leq
c_{i}\cdot\frac{1}{2\pi}a_{i}{}exp\left(-\frac{b_{i}}{4}|\hat{x_{i}}|^{2}\right)
 \label{6b1} \\ &\leq
d_{i}c_{i}\cdot\frac{1}{2\pi}a_{i}{}exp\left(-\frac{b'_{i}}{4}|\hat{x_{i}}|^{2}\right)
\label{6b2} \\ &\leq
d_{i}c_{i}\cdot\frac{1}{2\pi}a_{i}{}exp\left(-\frac{b_{min}}{4}|\hat{x_{i}}|^{2}\right)
\label{6b}
\end{align}
where $c_{i}=t_{1}$, $t_{2}$, or $t_{3}$. (\ref{6b2}) follows from
(\ref{6b1}) by using the same techniques as above and $d_{i}$ is a
constant independent of $P_{i}$. $b_{min}=\min(b'_{i})$, where
$b'_{i}=b_{i}$ for $b_{i}>0$; $b'_{i}=a_{i}$ for $b_{i}=0$;
$b'_{i}=-b_{i}$ for $b_{i}<0$. The average pairwise error
probability given the channel realization is
\begin{equation}
\overline{P}\left( X(i) \to X(j), i \neq j|{\mathbf{H}} = H
\right)\leq \tilde{K} \left( {\prod\limits_{i = 1}^m {a_i } } \right)^l \det
\left( {{b_{min}I} + H  H^\dag} \right)^{ - l} \label{7b}
\end{equation}
where $\tilde{K}$ is a constant which is not important here.
At a data rate $R=r\log{SNR}$, we have a total of
${SNR}^{lr}$ codewords. Applying the union bound, we have
\begin{align}
P\left( {error|{\mathbf{H}} = H} \right) \leq &\tilde{K}{SNR}^{lr}
\left( {\prod\limits_{i = 1}^m {a_i } } \right)^l \det \left(
{{b_{min}I} + H  H^\dag} \right)^{ - l}
\notag\\
= &\tilde{K}{SNR}^{lr} \left( {\prod\limits_{i = 1}^m \frac{a_i
}{b_{\min }}  } \right)^l \det \left( {I + \frac{1}{b_{\min }}H
H^\dag}
\right)^{ - l}\notag\\
{\dot  \leq } &\tilde{K}'{SNR}^{lr} \det \left( {I + {SNR}H
H^\dag} \right)^{ -
l}\notag\\
= & \tilde{K}'{SNR}^{lr} \prod\limits_{i = 1}^{\min (m,n)} {\left( {1
+ {SNR}\lambda _i } \right)^{ - l} } \notag \\
\doteq &{SNR}^{ - l\left[ {\sum\nolimits_{i = 1}^{\min (m,n)}
{\left( {1 - \alpha _i } \right)^ +   - r} } \right]}\label{8b}
\end{align}
where $\lambda_{i}$ are the singular values of $H$ and
$\lambda_{i}={SNR}^{-\alpha_{i}}$. Equation (\ref{8b}) is
exactly the same as (19) of \cite{diva}. Following the remaining
steps in \cite{diva}, it can be shown that for $l\geq(m+n-1)$, the
D-MG tradeoff with PAPR constraints is achievable.

\bibliographystyle{IEEEtran}
\bibliography{IEEEabrv,reference}

\begin{thebibliography}{10}
\providecommand{\url}[1]{#1}
\csname url@samestyle\endcsname
\providecommand{\newblock}{\relax}
\providecommand{\bibinfo}[2]{#2}
\providecommand{\BIBentrySTDinterwordspacing}{\spaceskip=0pt\relax}
\providecommand{\BIBentryALTinterwordstretchfactor}{4}
\providecommand{\BIBentryALTinterwordspacing}{\spaceskip=\fontdimen2\font plus
\BIBentryALTinterwordstretchfactor\fontdimen3\font minus
  \fontdimen4\font\relax}
\providecommand{\BIBforeignlanguage}[2]{{%
\expandafter\ifx\csname l@#1\endcsname\relax
\typeout{** WARNING: IEEEtran.bst: No hyphenation pattern has been}%
\typeout{** loaded for the language `#1'. Using the pattern for}%
\typeout{** the default language instead.}%
\else
\language=\csname l@#1\endcsname
\fi
#2}}
\providecommand{\BIBdecl}{\relax}
\BIBdecl

\bibitem{diva}
L.~Zheng and D.~N.~C. Tse, ``Diversity and multiplexing: A fundamental tradeoff
  in multiple antenna channels,'' \emph{{IEEE} Trans. Inf. Theory}, vol.~49,
  no.~5, pp. 1073--1096, May. 2003.

\bibitem{acht}
H.~Yao and G.~W. Wornell, ``Achieving the full mimo diversity-multiplexing
  frontier with rotation based space-time codes,'' in \emph{Allerton Conf.
  Comm. Control and Computing.}, Oct 2003.

\bibitem{anop}
P.~Dayal and M.~K. Varanasi, ``An optimal two transmit antenna spacetime code
  and its stacked extension,'' in \emph{Asilomar Conf. on Signals, Systems and
  Computers, Monterey, CA}, Nov 2003.

\bibitem{appu}
S.~Tavildar and P.~Viswanath, ``Approximately universal codes over slow fading
  channels,'' \emph{Submitted to IEEE Trans. Inform. Theory}, Feb 2005.

\bibitem{theg}
J.-C. Belfiore, G.~Rekaya, and E.Viterbo, ``The golden code: a $2\times2$
  full-rate space-time code with non-vanishing determinants,'' \emph{{IEEE}
  Trans. Inf. Theory}, vol.~51, no.~4, pp. 1432--1436, April 2005.

\bibitem{latc}
H.~E. Gamal, G.~Caire, and M.~Damen, ``Lattice coding and decoding achieve the
  optimal diversity-multilpexing tradeoff of mimo channels,'' \emph{{IEEE}
  Trans. Inf. Theory}, vol.~50, pp. 968--985, June 2004.

\bibitem{pers}
P.~Elia, B.~Sethuraman, and P.~V. Kumar, ``Perfect space time codes with
  minimum and non-minimum delay for any number of antennas,'' \emph{preprint,
  submited to arXiv:cs.IT}, Dec 2005.

\bibitem{exps}
P.~Elia, S.~A. Pawar, K.~R. Kumar, P.~V. Kumar, and H.~Lu, ``Explicit
  space-time codes achieving the diversity-multiplexing gain tradeoff,''
  \emph{{IEEE} Trans. Inf. Theory}, vol.~52, no.~9, Sep 2006.

\bibitem{maxd}
P.~Dayal and M.~K. Varanasi, ``Maximal diversity algebraic space-time codes
  with low peak-to-mean power ratio,'' \emph{{IEEE} Trans. Inf. Theory},
  vol.~51, pp. 1961--1708, May 2005.

\bibitem{pear}
H.~Kwok, ``Shape up: peak-power reduction via constellation shaping,'' Ph.D.
  dissertation, University of Illinois at Urbana-Champaign, 2001.

\bibitem{onsy}
H.~J. Smith, ``On systems of linear indeterminate equations and congruences,''
  \emph{Philos. Trans. R. Soc. Lond.}, vol. 151, pp. 293--326, June 1861.

\bibitem{intc}
A.~Mobasher and A.~K. Khandani., ``Integer-based constellation shaping method
  for papr reduction in ofdm systems,'' \emph{{IEEE} Trans. Commun.}, vol.~54,
  pp. 119--127, Jan 2006.

\bibitem{intd}
C.~Hermite, ``Sur l¡¦introduction des variables continues dans la theorie des
  nombres,'' \emph{J. Reine Angew. Math.}, pp. 191--216, 1851.

\bibitem{como}
A.~Storjohann, ``Computation of hermite and smith normal forms of matrices,''
  Master's thesis, University of Waterloo, Canada, 1994.

\bibitem{matf}
P.~Hao and Q.~Shi, ``Matrix factorizations for reversible integer mapping,''
  \emph{{IEEE} Trans. Signal Process.}, vol.~49, pp. 2314--2324, Oct 2001.

\bibitem{cust}
P.~Hao, ``Customizable triangular factorizations of matrices,'' \emph{Linear
  Algebra and Its Applications}, vol. 382, pp. 135--154, May 2004.

\bibitem{sc-fdma}
H.~G. Myung, J.~Lim, and D.~J. Goodman, ``Single carrier {FDMA} for uplink
  wireless transmission,'' \emph{IEEEE Vehicular Technology Magazine}, vol.~1,
  pp. 30--38, Sept. 2006.

\bibitem{thei}
J.~G. Smith, ``On the information capacity of peak and average power
  constrained gaussian channels,'' Ph.D. dissertation, Univ. of California,
  Berkeley, CA, 1969.

\bibitem{thec}
S.~Shamai and I.~Bar-David, ``The capacity of average and peak powerlimited
  quadrature gaussian channels,'' \emph{{IEEE} Trans. Inf. Theory}, vol.~41,
  no.~4, pp. 060--1071, 1995.

\bibitem{capo}
I.~E. Telatar, ``Capacity of multi-antenna gaussian channels,'' \emph{Tech.
  Rep:Bell Labs, Lucent Technologies}, 1995.

\bibitem{lays}
G.~J. Foschini, ``Layered space-time architecture for wireless communication in
  a fading environment when using multi-element antennas,'' \emph{Bell Labs
  Tech.}, vol.~1, no.~2, pp. 41--59, 1996.

\bibitem{mulc}
G.~D. {Forney Jr.} and L.-F. Wei, ``Multidimensional constellations {I}:
  Introduction, figures of merit, and generalized cross constellations,''
  \emph{{IEEE} J. Sel. Areas Commun.}, pp. 877--892, Aug 1989.

\bibitem{tres}
G.~D.Forney, ``Trellis shaping,'' \emph{{IEEE} Trans. Inf. Theory}, vol.~38,
  pp. 281--300, Mar 1992.

\bibitem{sham}
A.~K. Khandani and P.~Kabal, ``Shaping multidimensional signal spaces- part
  {I}: Optimum shaping, shell mapping,'' \emph{{IEEE} Trans. Inf. Theory}, pp.
  1794--1808, Nov 1993.

\bibitem{onop}
R.~Laroia, N.~Favardin, and S.~Tretter, ``On optimal shaping of
  multidimensional constellations,'' \emph{{IEEE} Trans. Inf. Theory}, June
  1994.

\bibitem{intf}
S.~Oraintara, Y.-J.Chen, and T.~Q.Nguyen, ``Integer fast fourier transform,''
  \emph{{IEEE} Trans. Signal Process.}, vol.~50, pp. 607--618, Mar 2002.

\bibitem{newn}
F.~A.~M. L.Bruckens and A.~W.~M. van~den Enden, ``New networks for perfect
  inversion and perfect reconstruction,'' \emph{{IEEE} J. Sel. Areas Commun.},
  vol.~10, 1992.

\bibitem{perb}
F.~Oggier, G.~Rekaya, J.-C. Belfiore, and E.~Viterbo, ``Perfect space-time
  block codes,'' \emph{Submitted to IEEE Trans. Inform. Theory}, 2006.

\bibitem{aunified}
A.~Murugan, H.~El~Gamal, M.~Damen, and G.~Caire, ``A unified framework for tree
  search : Rediscovering the sequential decoder,'' \emph{{IEEE} Trans. Inf.
  Theory}, vol.~52, no.~3, pp. 933--953, Mar. 2006.

\end{thebibliography}

\end{document}